\documentclass[12pt,a4paper,twoside]{article}
\usepackage{hyperref}
\usepackage{upgreek}
\usepackage{textgreek}
\usepackage{amssymb,amsmath,amsthm}
\usepackage{cite}
\usepackage{tikz}
\usepackage[utf8]{inputenc}
\usetikzlibrary{arrows,backgrounds,decorations.pathmorphing,decorations.pathreplacing,positioning,fit}
\usepackage{enumerate}
\usepackage{multicol}
\usepackage{accents}
\usepackage{stackengine}
\usepackage[conditional,light,first,bottomafter]{draftcopy}
\draftcopyName{DRAFT\space\today}{130}
\draftcopySetScale{65}
\usepackage{lmodern}
\usepackage{microtype}
\usepackage{graphicx}
\usepackage{subcaption}
\usepackage{dsfont}
\usepackage{colonequals}

\usepackage{pstricks}
\usepackage{yhmath}
\usepackage{fullpage}
\usepackage{setspace}
\makeatletter
\renewcommand{\section}{\@startsection%
{section}%
{1}%
{0em}%
{1.7em}%
{1.2em}%
{\normalfont\large\centering\bfseries}}
\renewcommand{\@seccntformat}[1]%
{\csname the#1\endcsname.\hspace{0.5em}}
\makeatother

\numberwithin{equation}{section}
\newtheorem{thm}{Theorem}[section]
\newtheorem{prop}[thm]{Proposition}
\newtheorem{lem}[thm]{Lemma}
\newtheorem{cor}[thm]{Corollary}
\theoremstyle{defn}
\newtheorem{defn}[thm]{Definition}
\newtheorem{rem}[thm]{Remark}

\newcommand{\abs}[1]{\left|#1\right|}
\newcommand\rE[1]{\upharpoonleft_{#1}}
\newcommand\lrb[1]{\left\lbrace#1 \right\rbrace}
\newcommand\lrp[1]{\left(#1 \right)}
\newcommand\C{{\mathbb C}}
\newcommand\N{{\mathbb N}}
\newcommand\R{{\mathbb R}}
\newcommand\cc[1]{\overline {#1}}
\newcommand\vb[1]{\!\left\langle {#1} \right |}
\newcommand\vk[1]{\left |{#1}\right\rangle\!}
\newcommand\ip[2]{\left\langle {#1},{#2} \right\rangle}
\newcommand\no[1]{\left\| {#1} \right\|}
\newcommand\tr[1]{{{\rm tr}}\left(#1\right)}
\newcommand\dom{\mathrm dom\,}
\newcommand\ran{\mathrm ran\,}
\newcommand\Span{\mathrm span\,}

\makeatletter
\newcommand{\mathleft}{\@fleqntrue\@mathmargin0pt}
\newcommand{\mathcenter}{\@fleqnfalse}
\makeatother
\begin{document}
\begin{titlepage}
\title{The fast recurrent subspace on an\\ $N$-level quantum energy transport model
}
\author{
\normalsize\textbf{Jorge R. Bola\~nos-Serv\'in, Josu\'e I. Rios-Cangas and Alfredo Uribe}
\\
\small Departamento de Matem\'aticas \\[-1.6mm] 
\small Universidad Aut\'onoma Metropolitana, Iztapalapa Campus \\[-1.6mm] 
\small San Rafael Atlixco 186, 09340 Iztapalapa, Mexico City.\\[-1.6mm]
\small \texttt{jrbs@xanum.uam.mx}, \;
\small \texttt{jottsmok@xanum.uam.mx}, \; 
\small \texttt{alur@xanum.uam.mx}
}

\date{\today}
\maketitle

\begin{center}
\begin{minipage}{5in}
\centerline{{\bf Abstract}} \medskip
The fast recurrent subspace (the biggest support of all invariant states) of a Weak Coupling Limit Type Quantum Markov Semigroup modeling a quantum transport open system of $N$-energy levels is determined. This is achieved by characterizing the structure of all the invariant state and their spectra in terms of a natural generalization of the Discrete Fourier Transform operator. Finally, the attraction domains and long-time behavior of the evolution are studied on hereditary subalgebras where faithful invariant states exist.\end{minipage}
\end{center}
\thispagestyle{empty}
\end{titlepage}

\section{Introduction}\label{sec:intro}\setcounter{equation}{0}
The family of Quantum Markov Semigroups (QMS's) is a tool which can be used to model of the evolution without memory of a microscopic system in accordance with the laws of quantum physics in the framework of open quantum systems. From a mathematical point of view, QMS's are a natural generalization of classical Markov semigroups on a function spaces in classical probability to a non-commutative operator algebras. This generalization gives a rigorous basis to the study of the qualitative behavior of evolution equations (master equations) on an operator
algebra, which can be computed explicitly in some cases or simulated numerically (see \cite{MR3804452} and the references therein).
	
As such, concepts like irreducibility, transience, and recurrence have been defined as the natural extension of the corresponding classical ones, for instance irreducible semigroups are shown to be either transient or recurrent \cite{SGFGOC2011}. A QMS is shown to be decomposable into “sub”-semigroups corresponding to classes of transient and recurrent states through the \textit{fast recurrent projection} $P_{\mathcal R_\mathcal L}$  \cite{MR2214906}, where  the \emph{fast recurrent subspace} $\mathcal R_\mathcal L$ is determined by the supports of normal invariant states. Determining the fast recurrent space allows restricting the domain of the semigroup to interesting hereditary subalgebras, where the  faithful invariant states exist and  long-time asymptotic properties are exhibited \cite{MR661704,MR479136}. 

This paper is a follow-up to the question left open in \cite{MR4192523} where the structure of the invariant states supported on some subspace $V$ of a quantum transport model of $N$-levels was determined.  This model is formulated in terms of a GKSL generator $\mathcal L$  of a weak coupling limit type QMS (WCLT QMS), where every Kraus operator is seen as a scalar multiple of a linear transformation, namely a \emph{transition operator}, which naturally generalizes the discrete Fourier transform between two Hilbert spaces. The transition operators play a fundamental role in the description of all the invariant states. The structure of the invariant states is attained by means of the powers of the transport operator (the orthogonal sum of all transition operators). The invariant states shed light on the so-called detailed balance (see Remark~\ref{eq:detail-balance-state2}) which is crucial in the study of ergodic QMS's and its  relative entropy \cite{MR3666729}.

The model discussed here generalizes the original setting in \cite{MR3399653} as well as the variations presented in \cite{MR3860251,MR4107240,AGQ2019D}. The main purpose of this paper is to prove the validity of the conjecture in \cite{MR4192523} and its consequences on the ergodic behavior of the QMS's. The conjecture establishes that the \emph{fast recurrent subspace}  satisfies
\begin{gather}\label{eq:frs-conjecture}
\mathcal R_\mathcal L = V \oplus  \text{\{one-dimensional   subspace\}. }\ 
\end{gather}
To this end, we first characterize all the invariant state of $\mathcal L$ and their supports to obtain \eqref{eq:frs-conjecture}.  Thereby, the transport scheme of invariant states proved in \cite{MR4192523} for some invariant states holds for any invariant state. This scheme establishes that any (non-trivial) invariant state is in fact a state supported on a smaller subspace of the first level, which is then transported along the rest of the levels. A similar transportation scheme is proved for the spectrum of the invariant states. With the above, we are able to explicitly study the long-term behavior and attraction domains of states under a suitable dimension hypothesis.

The structure of the paper is as follows: we briefly recall in Section~\ref{s2} some standard properties of the transition and transport operators  presented in \cite{MR4192523}. We define in Section~\ref{s3} the $N$-level WCLT QMS and describe the quantum transport model. The fast recurrent subspace is addressed in Section~\ref{s4} and we prove here the conjecture \eqref{eq:frs-conjecture} (see Theorem~\ref{th:fast-recurrent-subspace}). We also emphasize here the characterization of all invariant states, which is provided by Theorem~\ref{th:convex-decomposition-IS}. The study of the spectrum of any invariant state is addressed in Section~\ref{s6}, which is described by Theorem~\ref{cor:spectrum-decomposition} as a convex combination of spectra of their states in the first level. Section~\ref{s5} is devoted to the study of attraction domains and the long-time asymptotics of hereditary semigroups acting on hereditary subalgebras associated with subspaces of the first level. This will permit us to describe the evolution of states in certain subalgebras, in terms of structures of invariant states (see Theorems~\ref{cor:evolution-rho} and \ref{th:evolution-states-vw}). To conclude and as illustrative examples of this work, we present in Section~\ref{s7} the Kozyrev-Volovich \cite{Kozyrev2018} and Aref’eva-Volovich-Kozyrev \cite{MR3399653} quantum photosynthesis models.

\section{Transition and transport operators}\label{s2}\setcounter{equation}{0}
For $N\in\N$ let us consider a finite-dimensional Hilbert space $\mathcal H=\bigoplus_{k=0}^{N+1} E_k$, divided into $n_k$-dimensional mutually orthogonal subspaces $E_k$,  each one with \emph{canonical} basis
\begin{gather}\label{eq:k-canonical-bases}
\lrb{\vk{a_k}\,:\,0\leq a\leq n_k-1}\,,
\end{gather} 
where  $n_k\geq n_{k+1}$ and $n_0=n_{N+1}=1$ (see Fig. \ref{fig:graph}). For simplicity, the orthogonal projection of $\mathcal H$  onto $E_k$ shall be denoted by $P_k$, while for any other subspace $M\subset \mathcal H$, $P_M$ denotes the orthogonal projection onto $M$.

In contrast to the basis \eqref{eq:k-canonical-bases}, we also consider the \emph{entangled} basis
\begin{gather*}
\lrb{\varphi_{a_k}\,:\,\,0\leq a\leq n_k-1}\,,\quad\text{where}\quad \varphi_{a_k}\colonequals\frac{1}{\sqrt{n_k}}\sum_{b=0}^{n_k-1}\zeta_k^{-ab}\vk{b_k}\,, 
\end{gather*}
with $\zeta_k\colonequals e^{2\pi i/n_k}$, which is an orthonormal basis on $E_k$, for $k=0,\dots,N+1$.

\begin{defn}\label{def:transition-operator} For $k=0,\dots,N$, the \emph{transition operator} $Z_k\colon E_k\to E_{k+1}$ is given by
\begin{gather}\label{eq:Gen-discreteFF}
Z_{k}\colonequals\frac{1}{\sqrt{n_k}}\sum_{a=0}^{n_{k+1}-1}\sum_{b=0}^{n_k-1}\zeta_k^{ab}\vk{a_{k+1}}\vb{b_{k}}\,.
\end{gather}
\end{defn}
Note that $Z_0=\sqrt{n_1}\vk{\varphi_{0_1}}\vb{0_0}$. Thus, $\ker Z_0=E_0^\perp$ and $\ker Z_0^*=\{\varphi_{0_1}\}^\perp$. Besides,
\begin{gather*}
Z_k=\sum_{a=0}^{n_{k+1}-1}\vk{a_{k+1}}\vb{\varphi_{a_{k}}}\,, \quad k=1,\dots,N
\end{gather*}
which implies 
\begin{gather}\label{eq:Zentangled-canonical}
Z_{k}\varphi_{a_k}=\vk{a_{k+1}}\quad\mbox{and}\quad Z_{k}^*\vk{a_{k+1}}=\varphi_{a_k}\,.
\end{gather}
In addition, 
\begin{gather*}
\ker Z_k=\lrb{\Span\lrb{\varphi_{a_k}}_{a=0}^{n_{k+1}-1}}^\perp\quad\mbox{and}\quad \ker Z_k^*=E_{k+1}^\perp\,.\quad (k=1,\dots,N)  
\end{gather*}
Thereby, it is a simple matter to verify the following properties (cf. \cite{MR4192523}):
\begin{enumerate}
\item $Z_0Z_0^*=n_1 P_{\varphi_{0_1}}$ and $Z_0^*Z_0=n_1 P_{0}$.
\item For $k=1,\dots,N$, it follows that $Z_kZ_k^*=P_{k+1}$ and 
\begin{gather*}
\abs Z_k\colonequals Z_k^*Z_k=\sum_{a=0}^{n_{k+1}-1}\vk{\varphi_{a_{k}}}\vb{\varphi_{a_{k}}}\,,
\end{gather*}
which is a subprojection of $P_k$. Besides, $\ker \abs Z_k=\ker Z_k$. 
\item The last item implies that ${Z_k}$ and $Z_k^*$ are isometric isomorphisms between the subspaces $\abs Z_{k}E_k$ and $E_{k+1}$.
\end{enumerate}

It is useful to consider the orthogonal projection onto $\ker \abs Z_k$, given by
\begin{gather*}
\abs Z_k^\perp\colonequals P_k-\abs Z_k=\sum_{a=n_{k+1}}^{n_k-1}\vk{\varphi_{a_{k}}}\vb{\varphi_{a_{k}}}\,,\qquad k=1,\dots,N\,.
\end{gather*}
Also, we regard the \emph{transport} operator 
\begin{gather}\label{eq:trasport-operator}
Z\colon \bigoplus_{k=0}^NE_k\to\mathcal H\quad\mbox{as}\quad Z\colonequals  \bigoplus_{k=0}^NZ_k\,,
\end{gather}
which satisfies $ZP_k=Z_k$, 
\begin{gather*}
ZZ^*=n_1P_{\varphi_{0_1}}\oplus\bigoplus_{k=2}^{N+1}P_k\quad\mbox{and}\quad Z^*Z=n_1P_0\oplus\bigoplus_{k=1}^{N}\abs Z_k\,.
\end{gather*}
Thus, the maps ${Z}$ and $Z^*$ are isometric isomorphisms between $\bigoplus_{k=1}^{N}\abs Z_{k}E_k$ and $\bigoplus_{k=1}^{N}E_{k+1}$.

It is clear from \eqref{eq:Zentangled-canonical} that
\begin{gather}\label{eq:transition-on-bases}
Z\varphi_{a_k}=\vk{a_{k+1}}\quad\mbox{and}\quad Z^*\vk{a_{k+1}}=\varphi_{a_k}\,,\quad k=1,\dots,N\,.
\end{gather}
Besides, for $k=1,\dots,N-1$ and $m=1,2,\dots\leq (N-k)/2$, it follows that (cf. \cite[Cor.\,4]{MR4192523})
\begin{align}\label{eq:0k-m-even}
\begin{split}
Z^{2m-1}\vk{0_k}&=\prod_{j=0}^{m-1}\left(\frac{n_{k+2j+1}}{n_{k+2j}}\right)^{1/2}\varphi_{0_{k+2m-1}}
\,;\\[3mm]
Z^{2m}\vk{0_k}&=\prod_{j=0}^{m-1}\left(\frac{n_{k+2j+1}}{n_{k+2j}}\right)^{1/2}\vk{0_{k+2m}}\,.
\end{split}
\end{align}

Both transitions \eqref{eq:Gen-discreteFF} and transport \eqref{eq:trasport-operator} operators play a crucial role in the sequel.

 \section{$N$-level quantum energy transport model}\label{s3}\setcounter{equation}{0}

Recall that in an open quantum system (a quantum system interacting with the environment) the evolution of a state $\rho \mapsto {\mathcal T}_{t}(\rho)$, $t\geq 0$, is described by completely positive trace-preserving maps  $\mathcal T_t$ and the master equation  
\begin{gather*}
\frac{d{\mathcal T}_{t}(\rho)}{dt}= {\mathcal L}\lrp{{\mathcal T}_{t}(\rho)}\,, \qquad {\mathcal T}_{0}(\rho)=\rho
\end{gather*}
which involves an infinitesimal generator ${\mathcal L}$ with the Gorini-Kossakowski-Sudarshan and Lindblad (GKSL) structure.  The family $({\mathcal T}_{t})_{t\geq 0}$ of operators acting on $L_{1}({\mathcal H})$ (the space of finite trace operators) is called \textit{Quantum Markov Semigroup} (QMS).

We consider a GKSL Markov generator ${\mathcal L}$ belonging to the class of Weak Coupling Limit Type (WCLT) with degenerate reference Hamiltonian
\begin{gather*}
H\colonequals \sum_{k=0}^{N+1} \varepsilon_kP_k \,,\qquad (P_k\mbox{ the orthogonal projection of $\mathcal H$  onto $E_k$})
\end{gather*}
where the positive energies satisfy $\varepsilon_k>\varepsilon_{k+1}$ and the positive \emph{Bohr frequencies} (q.v. \cite[Subsect.\,1.1.5]{MR1929788}) $\omega_k=\varepsilon_k-\varepsilon_{k+1}$ are assumed to be pairwise different, for $k=0,\dots,N$.

In this fashion, $H$ corresponds to a quantum graph, viz. a graph whose vertices are the canonical basis $\lrb{\vk{a_k}\,:\,0\leq a\leq n_k-1}_{k=0}^{N+1}$ of $\mathcal H=\bigoplus_{k=0}^{N+1} E_k$ (see Fig. \ref{fig:graph}), and edges 
\begin{gather*}
\lrb{\zeta_k^{ab}\,:\, a=0,\dots,n_k-1\mbox{ and }b=0,\dots,n_{k+1}-1}_{k=0}^N\,,
\end{gather*}
where edge $\zeta_k^{ab}$ connects $\vk{a_{k}}$ with $\vk{b_{k+1}}$. 
\begin{figure}[h]
\centering
\begin{tikzpicture}
  [scale=.75,auto=left,every node/.style={}]
  \node (n0) at (4.5,7.5) {$0_0$};
  \node (n1) at (0,6)  {$0_1$};
  \node (n2) at (3,6)  {$1_1$};
  \node (n3) at (6,6) {$\cdots$};
  \node (n4) at (9,6)  {$(n_1-1)_1$};
   \node (m1) at (0,4.5)  {$0_2$};
   \node (m2) at (3,4.5)  {$1_2$};
   \node (m3) at (6,4.5)  {$\cdots$};
   \node (m4) at (9,4.5)  {$(n_2-1)_2$};
   \node (p1) at (0,3)  {$\vdots$};
   \node (p2) at (3,3)  {$\vdots$};
   \node (p3) at (6,3)  {$\ddots$};
   \node (p4) at (9,3)  {$\vdots$};
    \node (N1) at (0,1.5)  {$0_N$};
   \node (N2) at (3,1.5)  {$1_N$};
   \node (N3) at (6,1.5)  {$\cdots$};
   \node (N4) at (9,1.5)  {$(n_N-1)_N$};
    \node (o0) at (4.5,0)  {$0_{N+1}$};
     \node (lM) at (13,7.5)  {$0$-level\,,\quad $E_0$};
      \node (l1) at (13,6)  {$1$-level\,,\quad $E_1$};
       \node (l2) at (13,4.5)  {$2$-level\,,\quad $E_2$};
        \node (ld) at (13,3)  {$\vdots$};
         \node (lN) at (13,1.5)  {$N$-level\,,\quad $E_N$};
          \node (lm) at (13,0)  {$N+1$-level\,,\quad $E_{N+1}$};
  \foreach \from/\to in {n0/n1,n0/n2,n0/n3,n0/n4,o0/N1,o0/N2,o0/N3,o0/N4,n1/m1,n1/m2,n1/m3,n2/m1,n2/m2,n2/m3,n3/m1,n3/m2,n3/m3,n3/m4,n4/m3,n4/m4,
  m1/p1,m1/p2,m1/p3,m2/p1,m2/p2,m2/p3,m3/p1,m3/p2,m3/p3,m3/p4,m4/p3,m4/p4,
  p1/N1,p1/N2,p1/N3,p2/N1,p2/N2,p2/N3,p3/N1,p3/N2,p3/N3,p3/N4,p4/N3,p4/N4}
    \draw (\from) -- (\to);
\end{tikzpicture}
\caption{Graph of states and transitions.}\label{fig:graph}
\end{figure}
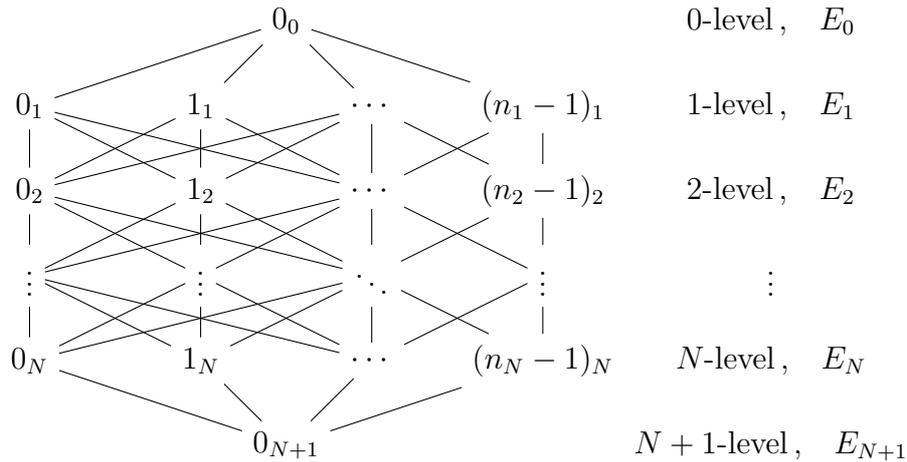

The structure of the WCLT Markov generator ${\mathcal L}$, in the Schr\"odinger's picture, is given by
\begin{eqnarray*}
\begin{aligned}
\mathcal L(\rho)\colonequals \sum_{k=0}^N & -i[\Delta_{\omega_k},\rho]+
  \left(L_{-, \omega_k} \rho L_{-, \omega_k}^{*}- \frac{1}{2} \lrb{L_{-, \omega_k}^* L_{-, \omega_k},\rho}
  \right)  \\ &+
 \lrp{L_{+, \omega_k} \rho L_{+, \omega_k}^*-\frac{1}{2}\lrb{ L_{+, \omega_k}^* L_{+, \omega_k}, \rho}}\,,
\end{aligned}
\end{eqnarray*}
while the dual generator is
\begin{eqnarray*}
\begin{aligned}
\mathcal L^*(x) \colonequals \sum_{k=0}^N& i[\Delta_{\omega_k},x]+\lrp{L_{-, \omega_k}^* x L_{-, \omega_k}-\frac{1}{2}\lrb{L_{-, \omega_k}^*L_{-, \omega_k}x}} \\ &+
\lrp{L_{+, \omega_k}^* x L_{+, \omega_k}-\frac{1}{2}\lrb{ L_{+, \omega_k}^*L_{+, \omega_k} x}}\,.
\end{aligned}
\end{eqnarray*}

Explicitly, the Kraus operators take the form  
\begin{gather}
\label{eq:Gamma-Krauss}
L_{-,\omega_k}=\sqrt{\Gamma_{-,\omega_k}}Z_k\quad L_{+,\omega_k}=\sqrt{\Gamma_{+,\omega_k}}Z_k^*\,, \quad 0\leq k\leq N\,,
\end{gather} 
with $\Gamma_{+,\omega_N}=0$
and the effective Hamiltonian is 
\begin{align}\label{eq:eff-Hamiltonian}
H_{\rm{eff}}\colonequals \sum_{k=0}^N \Delta_{\omega_k}\,,\quad\mbox{where}\quad \Delta_{\omega_k}= \gamma_{-,\omega_k}Z_k^*Z_k-
\gamma_{+,\omega_k}Z_kZ_k^*\,.
\end{align} 
The term $\gamma_{+,\omega_N}Z_NZ_N^*$ is absent from the effective Hamiltonian in the model studied in \cite{MR3860251}.

We use \eqref{eq:Gamma-Krauss} and \eqref{eq:eff-Hamiltonian} to rewrite $\mathcal L$ as the following:
\begin{align}\label{eq:L-generator-rewrite}
\begin{split}
\mathcal L(\rho)={}&\rho\lrp{n_1\eta_{-,\omega_0}P_0+n_1\overline{\eta}_{+,\omega_0}P_{\varphi_{0_1}}+\sum_{j=1}^{N}\eta_{-,\omega_j} \abs{Z}_j+\overline{\eta}_{+,\omega_j} P_{j+1}}\\&+
\lrp{n_1\overline{\eta}_{-,\omega_0}P_0 + n_1\eta_{+,\omega_0}P_{\varphi_{0_1}}+\sum_{j=1}^{N}\overline{\eta}_{-,\omega_j}\abs{Z}_j + \eta_{+,\omega_j}P_{j+1}}\rho  \\ & +
\Gamma_{-,\omega_N} Z_N \rho Z_N^*+\sum_{k=0}^{N-1}\Gamma_{-,\omega_k} Z_k \rho Z_k^* + \Gamma_{+,\omega_k} Z_k^*\rho Z_k \,,
\end{split}
\end{align}
where $\eta_{\pm,\omega_k}=-\dfrac{\Gamma_{\pm,\omega_k}}{2}+i\gamma_{\pm,\omega_k}$ and $\eta_{+,\omega_N}=i\gamma_{+,\omega_N}$. The expression~\eqref{eq:L-generator-rewrite} will be useful in the sequel for characterizing invariant states, i.e., states  which belong to  $\ker \mathcal L$. 

Let $e^{\beta_k}\colonequals\Gamma_{-,\omega_k}/\Gamma_{+,\omega_k}$,  with $\beta_k\colonequals{\omega_k\beta(\omega_k)}$, for $k=1,\dots,N-1$, and  
\begin{gather}\label{eq:varphi-operator}
\varphi(\rho)\colonequals n_1\lrp{\Gamma_{+,\omega_0}P_{\varphi_{0_1}}\rho P_{\varphi_{0_1}}+\eta_{+,\omega_0}P_{\varphi_{0_1}}\rho+\cc\eta_{+,\omega_0}\rho P_{\varphi_{0_1}}}\,.
\end{gather}

\begin{rem}\label{rm:varphi-commutes-rho}
Clearly, a state $\rho$ commutes with $P_{\varphi_{0_1}}$ if and only if $\varphi(\rho)=0$.
\end{rem}

\section{The fast recurrent subspace}\label{s4}\setcounter{equation}{0}
We address in this section the main goal of this article (see Theorem~\ref{th:fast-recurrent-subspace}), which assert the conjecture proposed in \cite{MR4192523}. Let us start with the following.

\begin{defn}
The fast recurrent subspace $\mathcal R_{\mathcal L}$ of a GKSL is the biggest support of its invariant states, namely
\begin{gather}\label{eq:frs-RL}
\mathcal R_{\mathcal L}\colonequals\sup \lrb{{\rm supp}\, \rho\,:\, \rho\mbox{ is an invariant state}}\,.
\end{gather}
\end{defn}

Since any convex combination of invariant states is an invariant state, one has that $\mathcal{R_L}$ is the union of the ranges of the all invariant states of $\mathcal L$. Thereby, we shall investigate some properties that invariant states must possess in terms of their support, commutation with some projections and general transport structures.

The following result shows that any invariant states can be described on each level in terms of the transport of \eqref{eq:varphi-operator}.  

\begin{lem}\label{lem:invariant-states}
An invariant state $\rho$ commutes with $\abs Z_1$, $P_0$, $P_1,\dots,P_{N+1}$, satisfies $Z_N\rho Z_N^*=0$, $P_{\varphi_{0_1}}\rho P_{\varphi_{0_1}}=\lrp{n_1\Gamma_{+,\omega_0}}^{-1}\Gamma_{-,\omega_0}Z_0\rho Z_0^*$, and
\begin{gather}\label{eq:detail-balance-state0}
\rho P_{k+1}=e^{\beta_k}Z_k\rho Z_k^*-\frac{1}{\Gamma_{+,\omega_k}}Z^k\varphi(\rho){Z^*}^k\,,\quad k=1,\dots, N-1\,.
\end{gather}
\end{lem}
\begin{proof}
Since $\rho$ is an invariant state, then $\mathcal L(\rho)=0$. Thus, by virtue of \eqref{eq:L-generator-rewrite}, 
\begin{gather}\label{eq:commutep0}
0=P_1\mathcal L(\rho) P_0=\lrp{n_1\eta_{-,\omega_0}P_1+n_1\eta_{+,\omega_0}P_{\varphi_{0_1}}+\cc\eta_{-,\omega_1}\abs Z_1}\rho P_0\,,
\end{gather}
i.e., $P_1\rho P_0=0$, since $\abs Z_1=\oplus_{a=0}^{n_2-1}P_{\varphi_{a_1}}$, $P_k=\abs Z_k\oplus\abs Z_k^\perp$, for $k=1,\dots, N$, and the real parts of the coefficients of \eqref{eq:commutep0} are strictly negative. Besides,
\begin{gather*}
0=P_j\mathcal L(\rho) P_0=\lrp{\lrp{n_1\eta_{-,\omega_0}+\eta_{+,\omega_{k-1}}}P_j+\cc\eta_{-,\omega_{k}}\abs Z_j}\rho P_0\,,
\end{gather*}
implies $P_j\rho P_0=0$, for $j=2,\dots N+1$. One can follow the same  reasoning from above to show that $P_j\rho P_k=0$, for all $j,k=0,\dots,N+1$, with $j\neq k$. Thereby,
\begin{gather*}
[\rho,P_k]=\sum_{s=0,\,s\neq k}^{N+1} P_s\rho P_k-P_k\rho\sum_{j=0,\,j\neq k}^{N+1} P_j=0\,,
\end{gather*}
which means that $\rho$ commutes with $P_k$ for all $k=0,\dots,N+1$. Now,
\begin{gather*}
0=\abs Z_1\mathcal L(\rho) \abs Z_1^\perp=n_1\eta_{+,\omega_0}P_{\varphi_{0_1}}\rho \abs Z_1^\perp+\cc\eta_{-,\omega_1} \abs Z_1\rho \abs Z_1^\perp\,,
\end{gather*}
yields $ \abs Z_1\rho \abs Z_1^\perp=0$ and since $\rho$ commutes with $P_1=\abs Z_1+\abs Z_1^\perp$,
\begin{gather}\label{eq:comuutes-Z1}
[\rho, \abs Z_1]=\rho P_1\abs Z_1- \abs Z_1P_1\rho=P_1\rho \abs Z_1- \abs Z_1\rho P_1=0\,,
\end{gather}
i.e., $\rho$ commutes with $\abs Z_1$. Note that $0=P_{N+1}\mathcal L(\rho)P_{N+1}$ implies ${Z_N\rho Z_N^*=0}$, and as consequence of $0=P_0\mathcal L(\rho)P_0$ one obtains $$P_{\varphi_{0_1}}\rho P_{\varphi_{0_1}}=\lrp{n_1\Gamma_{+,\omega_0}}^{-1}\Gamma_{-,\omega_0}Z_0\rho Z_0^*.$$ In addition, one computes from $0=\abs Z_1\mathcal L(\rho)\abs Z_1$ that
\begin{gather*}
Z_1^*\rho Z_1=e^{\beta_1}\abs Z_1\rho\abs Z_1-\frac{1}{\Gamma_{+,\omega_1}}\abs Z_1\varphi(\rho)\abs Z_1\,,
\end{gather*}
whence it follows \eqref{eq:detail-balance-state0} for $k=1$, since $Z_1Z_1^*=P_2$ and $Z_1^*Z_1=\abs Z_1$. Thus, if we suppose that \eqref{eq:detail-balance-state0} is true for $k-1$, viz. 
\begin{gather}\label{eq:h-induction}
\Gamma_{-,\omega_{k-1}}Z_{k-1}\rho Z_{k-1}^*=\Gamma_{+,\omega_{k-1}}\rho P_{k}+Z^{k-1}\varphi(\rho){Z^*}^{k-1}\,.
\end{gather}
Then, on account of $0=\abs Z_k\mathcal L(\rho)\abs Z_k$ and \eqref{eq:h-induction}, one has
\begin{align*}
0=&\, \Gamma_{+,\omega_k}Z_k^*\rho Z_k+\Gamma_{-,\omega_{k-1}}\abs Z_kZ_{k-1}\rho Z_{k-1}^* \abs Z_k -\lrp{\Gamma_{-,\omega_k}+\Gamma_{+,\omega_{k-1}}}\abs Z_k\rho \abs Z_k\\
=&\,\Gamma_{+,\omega_k}Z_k^*\rho Z_k+\abs Z_kZ^{k-1}\varphi(\rho){Z^*}^{k-1}\abs Z_k -\Gamma_{-,\omega_k}\abs Z_k\rho\abs Z_k\,, 
\end{align*}
which implies
\begin{gather}\label{eq:other-form}
Z_k^*\rho Z_k=e^{\beta_k}\abs Z_k\rho\abs Z_k-\frac{1}{\Gamma_{+,\omega_k}}\abs Z_kZ^{k-1}\varphi(\rho){Z^*}^{k-1}\abs Z_k\,,
\end{gather}
wherefrom one arrives at \eqref{eq:detail-balance-state0}, since $Z_kZ_k^*=P_{k+1}$ and $Z_k^*Z_k=\abs Z_k$. 
\end{proof}

\begin{rem}\label{rm:not-supported} For a state $\rho$ and $u\in\mathcal H$, one has that $\ip{u}{\rho u}=0$ if and only if $u\in\ker \rho$. Indeed, since $\rho$ is positive, if
\begin{gather*}
0=\ip{u}{\rho u}=\no{\rho^{1/2}u}^2\,,
\end{gather*}
then $\rho^{1/2}u=0$, i.e., $\rho u=\rho^{1/2}\rho^{1/2}u=0$. The converse is immediate.
\end{rem}

Now we establish necessity and sufficiency conditions for a state  $\rho$ to be invariant in terms of its support, transport operators, and commutation with projections. Indeed, we will see that in this case $\rho$ must commute with $P_{\varphi_{0_1}}$, which implies by Remark~\ref{rm:varphi-commutes-rho} that  $\varphi(\rho)=0$.

\begin{thm}\label{th:properties-invariant-state}
A state $\rho$ is invariant if and only if it is supported in $\lrb{\vk{0_0},\varphi_{0_1},\varphi_{0_N}}^\perp$, commutes with $P_1,\dots,P_{N+1}, \abs Z_1,\dots,\abs Z_{N-1}$, and  
\begin{gather}\label{eq:detail-balance-state}
\rho P_{k+1}=e^{\beta_k}Z_k\rho Z_k^*\,,\quad k=1,\dots,N-1\,.
\end{gather}
\end{thm}
\begin{proof}
If $\rho$ is invariant then it satisfies conditions of Lemma~\ref{lem:invariant-states}, which will be used freely. So, $0=Z_N\rho Z_N^*=\ip{\varphi_{0_N}}{\rho \varphi_{0_N}}P_{0_{N+1}}$ and Remark~\ref{rm:not-supported} imply $\varphi_{0_N}\in\ker \rho$ and $\rho P_{\varphi_{0_N}}=0$. Besides, since $\abs Z_N=P_{\varphi_{0_N}}$ and $\rho$ commutes with $P_N$, then $0=\mathcal L(\rho)P_{N}$ and \eqref{eq:detail-balance-state0} imply
\begin{gather}\label{eq:varphi-transported}
\rho P_N=e^{\beta_{N-1}}Z_{N-1}\rho Z_{N-1}^*\qquad\mbox{and}\qquad Z^{N-1}\varphi(\rho){Z^*}^{N-1}=0\,.
\end{gather}
Now, for $k=2,\dots,N-1$, one computes from $0=Z_k\mathcal L(\rho)Z_k^*$ that
\begin{gather}\label{eq:projectionsZk}
\rho P_{k+1}=\lrp{e^{\beta_k}+\frac{\Gamma_{+,\omega_{k-1}}}{\Gamma_{+,\omega_k}}}Z_k\rho Z_k^*-\frac{\Gamma_{-,\omega_{k-1}}}{\Gamma_{+,\omega_k}}Z_kZ_{k-1}\rho Z_{k-1}^*Z_k^*\,.
\end{gather}
For $j=0,\dots,N-3$, we claim that 
\begin{align}\label{eq:recursive-Z-powers}
\begin{split}
&Z_{N-1}Z_{N-2}\cdots Z_{N-1-j}\rho Z_{N-1-j}^*\cdots Z_{N-2}^*Z_{N-1}^*\\&=e^{\beta_{N-2-j}}Z_{N-1}Z_{N-2}\cdots Z_{N-2-j}\rho Z_{N-2-j}^*\cdots Z_{N-2}^*Z_{N-1}^*\,.
\end{split}
\end{align}
Indeed, the left-hand side of \eqref{eq:varphi-transported} and \eqref{eq:projectionsZk} imply the case $j=0$ in \eqref{eq:recursive-Z-powers}. Thus, by induction, we may suppose that \eqref{eq:recursive-Z-powers} holds for $j-1$ and after substituting $P_{N-1-j}$ of \eqref{eq:projectionsZk} in
\begin{align*}
&Z_{N-1}\cdots Z_{N-1-j}\rho Z_{N-1-j}^*\cdots Z_{N-1}^*\\&=Z_{N-1}\cdots Z_{N-1-j}\rho P_{N-1-j} Z_{N-1-j}^*\cdots Z_{N-1}^*
\end{align*}
one obtains \eqref{eq:recursive-Z-powers}. 
Thus, we use \eqref{eq:recursive-Z-powers} recursively in the left-hand side of \eqref{eq:varphi-transported} to get
\begin{gather}
\rho P_{N}=e^{\sum_{j=1}^{N-1}\beta_j}Z^{N-1}\rho {Z^*}^{N-1}\,.
\end{gather}
Besides, since $0=\ip{\varphi_{0_N}}{\rho P_N\varphi_{0_N}}=e^{\sum_{j=1}^{N-1}\beta_j}\ip{{Z^*}^{N-1}\varphi_{0_N}}{\rho {Z^*}^{N-1}\varphi_{0_N}}$, then Remark~\ref{rm:not-supported} asserts that $ {Z^*}^{N-1}\varphi_{0_N}\in\ker\rho$. One has by \eqref{eq:0k-m-even} that 
\begin{gather*}
\ip{\varphi_{0_1}}{{Z^*}^{N-1}\varphi_{0_N}}=\ip{Z^{N-1}\varphi_{0_1}}{\varphi_{0_N}}=\ip{Z^{N-2}\vk{0_2}}{\varphi_{0_N}}\neq0\,.
\end{gather*}
Thus, the right-hand side of \eqref{eq:varphi-transported} implies $0=\ip{{Z^*}^{N-1}\varphi_{0_N}}{\varphi(\rho){Z^*}^{N-1}\varphi_{0_N}}$ and
\begin{align*}
0=&\,\Gamma_{+,\omega_0}\ip{{Z^*}^{N-1}\varphi_{0_N}}{P_{\varphi_{0_1}}\rho P_{\varphi_{0_1}}{Z^*}^{N-1}\varphi_{0_N}}\\&+\eta_{+,\omega_0}\ip{{Z^*}^{N-1}\varphi_{0_N}}{P_{\varphi_{0_1}}\rho{Z^*}^{N-1}\varphi_{0_N}}\\&+\cc\eta_{+,\omega_0}\ip{{Z^*}^{N-1}\varphi_{0_N}}{\rho P_{\varphi_{0_1}}{Z^*}^{N-1}\varphi_{0_N}}\\=&\,\Gamma_{+,\omega_0}\abs{\ip{\varphi_{0_1}}{{Z^*}^{N-1}\varphi_{0_N}}}^2\ip{\varphi_{0_1}}{\rho \varphi_{0_1}}\,,
\end{align*}
which implies $\ip{\varphi_{0_1}}{\rho \varphi_{0_1}}=0$, i.e., $\varphi_{0_1}\in\ker \rho$. Moreover, 
\begin{gather*}
0=P_{\varphi_{0_1}}\rho P_{\varphi_{0_1}}=\frac{\Gamma_{-,\omega_0}}{n_1\Gamma_{+,\omega_0}}Z_0\rho Z_0^*=\frac{\Gamma_{-,\omega_0}}{\Gamma_{+,\omega_0}}\ip{\vk{0_0}}{\rho\vk{0_0}}P_{\varphi_{0_1}}
\end{gather*}
fulfills $\ip{\vk{0_0}}{\rho\vk{0_0}}=0$ and $\vk{0_0}\in\ker\rho$. So, $\rho$ has support in $\lrb{\vk{0_0},\varphi_{0_1},\varphi_{0_N}}^\perp$. Note that $\varphi(\rho)=0$, which from \eqref{eq:detail-balance-state0} one obtains \eqref{eq:detail-balance-state}. 

It remains to prove that $\rho$ commutes with $\abs Z_j$, for $j=2,\dots,N-1$. One readily checks that $0=\abs Z_j \mathcal L(\rho)\abs Z^\perp_j$ and \eqref{eq:detail-balance-state} gives $\abs Z_j \rho \abs Z^\perp_j=0$, whence analogously to \eqref{eq:comuutes-Z1},  the assertion follows. 

Conversely, note that \eqref{eq:detail-balance-state} is equivalent to
\begin{gather}\label{eq:detail-balance-state2}
Z_k^*\rho Z_k=e^{\beta_k} \rho \abs Z_k\,,\quad k=1,\dots,N-1\,.
\end{gather}
Hence, using the commutation conditions, the support of $\rho$, and replacing \eqref{eq:detail-balance-state} and \eqref{eq:detail-balance-state2} in \eqref{eq:L-generator-rewrite}, one gets that $\mathcal L(\rho)=0$.
\end{proof}

We have mentioned in the proof of Theorem~\ref{th:properties-invariant-state} that \eqref{eq:detail-balance-state} and \eqref{eq:detail-balance-state2} are equivalent. The following generalizes these conditions.

\begin{rem}\label{eq:detail-balance} Condition \eqref{eq:detail-balance-state} in Theorem~\ref{th:properties-invariant-state} can be replaced by 
\begin{align}\label{eq:invariant-property-DB}
\rho Z_k=e^{\beta_k} Z_k\rho\,,\quad k=1,\dots,N-1\,.
\end{align}
Indeed, by \eqref{eq:detail-balance-state2},
$\rho Z_k=\rho P_{k+1}Z_k=Z_kZ_k^*\rho Z_k=e^{\beta_k}Z_k\abs Z_k\rho=e^{\beta_k}Z_k\rho$.
\end{rem}
By virtue of \eqref{eq:Gamma-Krauss} and  \eqref{eq:invariant-property-DB}, it follows that 
\begin{gather*}
\rho L_{-,\omega_k}=e^{\beta_k} L_{-,\omega_k}\rho\quad\mbox{and}\quad 
L_{+,\omega_k}\rho=e^{\beta_k} \rho L_{+,\omega_k}\,,\quad k=1,\dots,N-1
\end{gather*}
which is known as \emph{detailed balance} \cite{MR3666729} (c.f. \cite[Sect.\,3.2]{MR4107240}).

It is convenient to consider the \emph{interaction-free} subspace
 \begin{gather*}
 W\colonequals \bigcap_{k=0}^N\lrp{\ker L_{\pm, \omega_k} \cap \ker L^*_{\pm, \omega_k}}\,,
 \end{gather*}
which satisfies (cf. \cite{MR4192523}) 
\begin{align}\label{eq:interaction-freeW}
W=\bigcap_{k=0}^N\ker Z_k\cap \ker Z_k^*=P_1\ker Z_1=\Span\lrb{\varphi_{a_1}}_{a=n_{2}}^{n_{1}-1}\,.
\end{align}

\begin{rem}\label{rm:trivial-is}
Taking into account \eqref{eq:interaction-freeW} and \eqref{eq:L-generator-rewrite} one has that any state supported in $W$ is invariant. Besides, one readily checks that $P_{N+1}$ is  an invariant state as well.
\end{rem} 

The following shows a characterization of the support of invariant states.

\begin{cor}\label{cor:convex-decomposition} 
A state $\rho$ is invariant if and only if there exist invariant states $\eta,\tau$ supported in $W, W^\perp\ominus\{\vk{0_{N+1}}\}$, respectively, and scalars $\alpha, \beta, \lambda\geq0$, with $\alpha+\beta+\lambda=1$,
 such that 
\begin{gather}\label{eq:convex-decomposition}
\rho=\alpha \tau+\beta\eta+\lambda P_{N+1}\,.
\end{gather}
\end{cor}
\begin{proof}
If $\rho$ is invariant then by Theorem~\ref{th:properties-invariant-state}, it commutes with $P_{N+1},P_1$ and $\abs Z_1$, which implies the commutation with $\abs Z_1^\perp$. Thereby, $\ran \abs Z_1^\perp=W$ and $\ran P_{N+1}=\C\vk{0_{N+1}}$ reduce $\rho$ and they are orthogonal. Hence, 
\begin{gather*}
\rho=\rho\rE{W}\oplus\rho\rE{W^\perp\ominus\{\vk{0_{N+1}}\}}\oplus\rho\rE{\C\vk{0_{N+1}}}\,,
\end{gather*}
which by a suitable normalization, one yields \eqref{eq:convex-decomposition}. Note that $\tau$ is invariant since $\rho,\eta$ and $P_{N+1}$ are. The converse assertion is straightforward. 
\end{proof}
Corollary~\ref{cor:convex-decomposition} means that any invariant state is decomposed into a convex combination of invariant states supported in $W, W^\perp\ominus\{\vk{0_{N+1}}\}$ and $\C\vk{0_{N+1}}$\,. To continue describing the support of invariant states, it is useful to consider the following subspace:
\begin{align}\label{eq:calV-harmonic}
\begin{split}
 V\colonequals\{ &Z^n\vk{0_0}, {Z^*}^n\vk{0_{N+1}},{Z^*}^{s_m}\varphi_{0_{2m+1}}\,: \\ & 0\leq n\leq N, \, 1\leq m  \leq (N-1)/2,\, 1\leq s_m\leq 2m{\}}^\perp\,.
\end{split}
\end{align}
\begin{cor}\label{cor:support-invariant-state} Any invariant state is supported in $V\oplus\C\vk{0_{N+1}}$\,.
\end{cor}
\begin{proof}
By virtue of Corollary~\ref{cor:convex-decomposition}, it suffices to show that if an invariant state $\rho$ is supported in $\{\vk{0_{N+1}}\}^\perp$ then so is in $V$. In this fashion, one has from Theorem~\ref{th:properties-invariant-state} that $\rho$ has support in $\lrb{\vk{0_0},\vk{0_{N+1}},\varphi_{0_1},\varphi_{0_N}}^\perp$ and due to \eqref{eq:invariant-property-DB}, there exists $\alpha_n>0$ such that 
\begin{gather*}
\rho Z^n\vk{0_0}=\rho Z^{n-1}\varphi_{{0_1}}=\alpha_nZ^{n-1}\rho\varphi_{{0_1}}=0\,,\qquad n=1,\dots,N\,.\end{gather*} Analogously, $\rho Z^{*n}\vk{0_{N+1}}=0$, since $\varphi_{0_N}=Z^*\vk{0_{N+1}}$. Now, from \eqref{eq:0k-m-even}, there exists $\alpha_m>0$ such that $Z^{2m}\varphi_{0_1}=Z^{2m-1}\vk{0_2}=\alpha_m\varphi_{0_{2m+1}}$. Thereby, again by \eqref{eq:invariant-property-DB}, it follows that $\rho{Z^*}^{s_m}\varphi_{0_{2m+1}}=\alpha_{s_m}{Z^*}^{s_m}Z^{2m}\rho\varphi_{0_1}=0$, with $\alpha_{sm}>0$, as required.
\end{proof}

Let us denote 
\begin{gather*}
V_1\colonequals P_1V=P_1\mathcal H\ominus\{\varphi_{0_1},Z^{*N-1}\varphi_{0_N}\}\,.
\end{gather*}
So, Corollary~\ref{cor:support-invariant-state} implies that $W\subset V_1$, since any state supported in $W\subset P_1\mathcal H$ is invariant. The following result is adapted from \cite[Ths.\,4 and 5]{MR4192523}.
\begin{lem}\label{rm:IS-inVmW}
Any state $\rho$ supported in $V\ominus W$ is invariant if and only if there exists a unique state $\tau$ supported in $V_1\ominus W$ such that
\begin{gather}\label{eq:characterisation-invariant-state-level1}
\rho=c_\rho \sum_{n=0}^{N-1}e^{\sum_{j=0}^{n}\beta_j}Z^n\tau Z^{*n}\,,\qquad (\beta_0=0)
\end{gather}
where $c_\rho=\tr{\rho\abs Z_1}$. In such a case one has that $\ran\rho={V}\ominus{W}$ if and only if $\ran \tau={V_1}\ominus{W}$ 
\end{lem}

Lemma~\ref{rm:IS-inVmW} asserts that there is a one-to-one correspondence between the states supported in $V_1\ominus W$ and the invariant states supported in $V\ominus W$. Besides, the number $c_\rho$ in \eqref{eq:characterisation-invariant-state-level1} acts as a normalization constant.

The following theorem gives a general structure of invariant states.

\begin{thm}\label{th:convex-decomposition-IS}
A state $\rho$ is invariant if and only if there exist states $\eta,\tau$ supported in $W,V_1\ominus W$, respectively, and $\alpha,\beta,\lambda\geq0$, with $\alpha+\beta+\lambda=1$, such that 
\begin{gather}\label{eq:convex-decomposition-IS}
\rho=\alpha c\sum_{n=0}^{N-1}e^{\sum_{j=0}^{n}\beta_j}Z^n\tau Z^{*n}+\beta\eta+\lambda P_{N+1}\,.\qquad (\beta_0=0)
\end{gather}
where $c=\alpha^{-1}\tr{\rho\abs Z_1}$, when $\alpha\neq0$. Besides, $\ran\rho P_{V\ominus W}={V}\ominus{W}$ if and only if $\ran \tau={V_1}\ominus{W}$.
\end{thm}
\begin{proof}
If $\rho$ is invariant then by Corollaries~\ref{cor:convex-decomposition} and \ref{cor:support-invariant-state},  there exist invariant states $\hat\rho,\eta$ supported in $V\ominus W,W$, respectively, and scalars $\alpha, \beta, \lambda\geq0$, with $\alpha+\beta+\lambda=1$, such that $\rho=\alpha\hat\rho+\beta\eta+\lambda P_{N+1}$. Thus, since $\hat\rho$ satisfies Lemma~\ref{rm:IS-inVmW}, one arrives at \eqref{eq:convex-decomposition-IS}. If $\alpha\neq0$, then $\alpha^{-1}\tr{\rho\abs Z_1}=\alpha^{-1}\tr{\alpha c\tau}=c$. The converse readily follows by Corollary~\ref{cor:convex-decomposition} and Lemma~\ref{rm:IS-inVmW}. Also,  Lemma~\ref{rm:IS-inVmW} implies that $\ran\rho P_{V\ominus W}=\ran\hat\rho=V\ominus W$ if and only if $\ran\tau=V_1\ominus W$. 
\end{proof}

\begin{rem} There is no invariant state supported in $\lrp{V_1\oplus\C\vk{0_{N+1}}}^\perp$, since otherwise, $\beta=\lambda=0$ and $\tau=0$ in \eqref{eq:convex-decomposition-IS}, i.e., $\rho=0$, a contradiction.
\end{rem}

Recall that a state is said to be \emph{extremal} if it cannot be decomposed as a non-trivial convex combination of two different states. On the other hand, a state is called \emph{invariant-extremal} if it is invariant and cannot be represented as a non-trivial convex combination of two different invariant states.

\begin{rem}\label{rm:invariant-extremal-states} Clearly, $P_{0_{N+1}}$ is an invariant-extremal state. Besides, a state $\rho$ supported in $W$ is invariant-extremal if an only if there exists a unit vector $w\in W$, such that $\rho=\vk w\vb w$, i.e., $\rho$ is a pure state. Furthermore, an invariant state $\rho$ supported in $V\ominus W$ is invariant-extremal if and only if $\tau$ in \eqref{eq:characterisation-invariant-state-level1} is a pure state supported in $V_1\ominus W$ (see for instance \cite[Lem.\,4]{MR4192523}).
\end{rem}

The following result is an immediate consequence of Theorem~\ref{th:convex-decomposition-IS} and Remark~\ref{rm:invariant-extremal-states} (c.f. \cite[Th.\,6]{MR4192523}).
\begin{cor}\label{cor:invariant-extremal}
A state $\rho$ is invariant-extremal if and only if one of the following conditions is true:
\begin{enumerate} 
\item $\rho=P_{0_{N+1}}$.
\item $\rho=\vk w\vb w$, where $w\in W$ is a unit vector.
\item There exists a vector $u\in{V_1}\ominus {W}$, with $\no u^2=\tr{\rho\abs Z_1}$, such that 
\begin{gather*}
\rho=\sum_{n=0}^{N-1}e^{\sum_{j=0}^{n}\beta_j}Z^n\vk u\vb u Z^{*n}\,.
\end{gather*}
\end{enumerate}
\end{cor}

Now, we are ready to prove the conjecture of \cite{MR4192523}.

\begin{thm}\label{th:fast-recurrent-subspace}
The fast recurrent subspace $\mathcal R_{\mathcal L}=V\oplus\C\vk{0_{N+1}}$\,.
\end{thm}
\begin{proof}
It is clear from Corollary~\ref{cor:support-invariant-state} that $\mathcal R_{\mathcal L}\subset V\oplus\C\vk{0_{N+1}}$. On the other hand, by virtue of Theorem~\ref{th:convex-decomposition-IS}, one obtains an invariant state with range equal to $V\oplus\C\vk{0_{N+1}}$\,, which concludes the assertion.
\end{proof}

\begin{rem}[Dark states] On quantum energy transport models, it is useful to consider the \emph{bright photonic} vector $\varphi_{0_1}$ and the \emph{photonic} vector \cite{Kozyrev2018} (see also \cite[Sect.\,3]{MR3860251} and \cite[Ex.\,3.2]{MR4107240})
$\psi=e^{i\theta}\varphi_{0_1}$, with $\theta\in(0,2\pi)\backslash\{\pi\}$. The corresponding pure states of these vectors coincide with the so-call \emph{bright pure state} $P_{\varphi_{0_1}}$. Besides, a \emph{dark state} is a state $\rho$ which is orthogonal to the bright pure state, with respect to the Hilbert-Schmidt inner product, i.e.,
\begin{gather*}
0=\tr{\rho P_{\varphi_{0_1}}}=\ip{\varphi_{0_1}}{\rho \varphi_{0_1}}\,.
\end{gather*}
In this fashion, Remark~\ref{rm:not-supported} asserts that a state is dark if and only if it has support in $\{\varphi_{0_1}\}^\perp$. Therefore, by virtue of Theorem~\ref{th:properties-invariant-state} one has that any invariant state is dark.
\end{rem}

\section{The spectrum of invariant states}\label{s6}\setcounter{equation}{0}
We will describe the spectrum of any invariant state in terms of the spectra of states supported in $V_1\ominus W$. We start by mentioning that a state $\rho$ has spectrum $\sigma(\rho)\subset[0,1]$, with the sum of its elements equal one. Besides, if $\rho$ is an invariant state then $(V\oplus\C\vk{0_{N+1}})^\perp\subset\ker \rho$, due to Corollary~\ref{cor:support-invariant-state}. Hence, the following holds.
\begin{prop}
If $\rho$ is and invariant state, then $0\in\sigma(\rho)$, with multiplicity at least $\dim V^{\perp}-1$.
\end{prop}
Clearly, the spectrum of the invariant state $P_{N+1}$ is $\sigma(P_{N+1})=\{0,1\}$. 
\begin{thm}\label{th:spectrum-decomposition}
For an invariant state $\rho$, there exists an invariant state $\tau$  supported in $V\ominus W$, a state $\eta$  supported in $W$ and $\alpha,\beta,\lambda\in[0,1]$, such that
\begin{gather}\label{eq:cv-Spectrum}
\sigma(\rho)=\alpha\sigma(\tau)\cup\beta\sigma(\eta)\cup\{0,\lambda\}\,,\quad\mbox{with}\quad \alpha+\beta+\lambda=1\,.
\end{gather}
\end{thm}
\begin{proof}
By virtue of Corollaries~\ref{cor:convex-decomposition}  and~\ref{cor:support-invariant-state}, any invariant state $\rho$ is decomposed into an orthogonal sum 
\begin{gather}\label{eq:convex-convination-S}
\rho=\alpha \tau\rE{V\ominus W}\oplus\beta\eta\rE{W}\oplus\lambda P_{N+1}\rE{\C\vk{0_{N+1}}}\,,
\end{gather}
where $\tau$ is an invariant state supported in $V\ominus W$, $\eta$ is a state supported in $W$ and $\alpha,\beta,\lambda\in[0,1]$, with $\alpha+\beta+\lambda=1$. Hence, \eqref{eq:convex-convination-S} implies \eqref{eq:cv-Spectrum}.
\end{proof}

Recall that $W\subset V_1\ominus W$. Besides, Lemma~\ref{rm:IS-inVmW} asserts that every invariant state supported in $V\ominus W$ is completely determined by a unique state supported in $V_1\ominus W$. The following result uses the structure \eqref{eq:characterisation-invariant-state-level1} of an invariant state.

\begin{lem}\label{th:spectrum-is-vw}
Let  $\tau$ be a state supported in $V_1\ominus W$. If $\{\lambda_k\}_{k=1}^m$ are the non-zero eigenvalues of $\tau$, with respective eigenvectors $\{u_k\}_{k=1}^m$. Then the non-zero eigenvalues of the invariant state
\begin{gather}\label{eq:esp-is-vw}
c \sum_{n=0}^{N-1}e^{\sum_{j=0}^{n}\beta_j}Z^n\tau Z^{*n}\,,\qquad (\beta_0=0\mbox{ and $c$ a normalization constant} )
\end{gather} 
are
\begin{gather}\label{eq:eigenvalues-eigenvectors-is}
\lrb{c\lambda_k\no{Z^nu_k}^2e^{\sum_{j=0}^{n}\beta_j}}_{k=1,n=0}^{m,N-1}\,,
\end{gather}
with respective eigenvectors (up to normalization) $\lrb{Z^nu_k}_{k=1,n=0}^{m,N-1}$.
\end{lem}
\begin{proof}
Since $\tau=\sum_{k=1}^m\lambda_k\vk{u_k}\vb{u_k}$, which substituting in \eqref{eq:esp-is-vw}, one has that
\begin{gather*}
c \sum_{n=0}^{N-1}e^{\sum_{j=0}^{n}\beta_j}Z^n\tau Z^{*n}=\sum_{n=0}^{N-1}\sum_{k=1}^m c\lambda_k e^{\sum_{j=0}^{n}\beta_j}\vk{Z^nu_k}\vb{Z^{n}u_k}\,.
\end{gather*}
Thus, $\{Z^nu_k\}_{k=1,n=0}^{m,N-1}$ are the distinct eigenvectors of the selfadjoint operator \eqref{eq:esp-is-vw}. Therefore, one gets  \eqref{eq:eigenvalues-eigenvectors-is}, since $\ip{Z^nu_k}{Z^ru_s}=\no{Z^nu_k}^2\delta_{nr}\delta_{ks}$.
\end{proof}
\begin{rem}\label{rm:NC-by-spectrum} Since \eqref{eq:esp-is-vw} is a state, the constant $c$ in Lemma~\ref{th:spectrum-is-vw} satisfies 
\begin{gather*}
c=\lrp{\sum_{k=1}^{m}\sum_{n=0}^{N-1}\lambda_k\no{Z^nu_k}^2e^{\sum_{j=0}^{n}\beta_j}}^{-1}\,.
\end{gather*}
\end{rem}
\begin{cor}\label{cor:ixe-eigenvalues}
For a unit vector $u\in V_1\ominus W$, it follows that
\begin{gather}
c \sum_{n=0}^{N-1}e^{\sum_{j=0}^{n}\beta_j}Z^n\vk u\vb u Z^{*n}
\end{gather}
is an invariant-extremal state with non-zero eigenvalues 
\begin{gather*}
\lrb{c \no{Z^nu}^2e^{\sum_{j=0}^{n}\beta_j}}_{n=0}^{N-1}\,,\quad \mbox{with respective eigenvectors }\, \lrb{Z^nu}_{n=0}^{N-1}\,,
\end{gather*}
where $c=\lrp{\sum_{n=0}^{N-1}\no{Z^nu}^2e^{\sum_{j=0}^{n}\beta_j}}^{-1}$.
\end{cor}
\begin{proof}
It is simple from Corollary~\ref{cor:invariant-extremal}, Lemma~\ref{th:spectrum-is-vw} and Remark~\ref{rm:NC-by-spectrum}.
\end{proof}

The following result is straightforward from Theorem~\ref{th:spectrum-decomposition} and Lemma~\ref{th:spectrum-is-vw}.
\begin{thm}\label{cor:spectrum-decomposition} If $\rho$ is an invariant state, then there exist states $\tau,\eta$,  supported in $V_1\ominus W,W$, respectively, $\alpha,\beta,\lambda\in[0,1]$ and $c>0$, such that
\begin{gather*}
\sigma(\rho)=\alpha\sigma(\eta)\cup\{0,\beta\}\bigcup_{n=0}^{N-1}\lambda c e^{\sum_{j=0}^{n}\beta_j}\sigma(\tau) \,,
\end{gather*}
with $\alpha+\beta+\lambda=1$, where the constant $c$ satisfies Remark~\ref{rm:NC-by-spectrum}.
\end{thm}

According to Theorem~\ref{cor:spectrum-decomposition}, the spectrum of the invariant states depends only on the spectra of their states in the first level.

\section{Approach to equilibrium and attraction domains on hereditary subalgebras}\label{s5}\setcounter{equation}{0}

It is convenient in this section to consider a stratification of the subspace \eqref{eq:calV-harmonic} given by $V=\bigoplus_{k=1}^N V_k$, where $V_k\colonequals P_kV$. According to \cite[Lem.\,2]{MR4192523}, the following holds.
\begin{lem}\label{lem:properties-trasport-operator}
For $k=1,\dots,N-1$ and $j=0,\dots,N-k$, it follows that $Z^jV_k=V_{k+j}$. Besides, 
\begin{gather*}
Z^kV=\bigoplus_{j=k+1}^{N}V_{j}\qquad\mbox{and}\qquad V=\bigoplus_{j=0}^{N-1}Z^jV_1\,.
\end{gather*}
Moreover,  the transitions $Z_k\colon\abs Z_k V\to V_{k+1}$ and $Z^*_k\colon V_{k+1}  \to \abs Z_k V$ are isometric isomorphisms.
\end{lem}

We consider the \emph{decoherence-free} subalgebra (df-algebra for short) for $\mathcal T$,
\begin{gather*}
\mathcal N(\mathcal T)\colonequals \lrb{x\in \mathcal{B(H)}\,:\, \mathcal T_t(x^*x)=\mathcal T_t(x)^*\mathcal T_t(x)\,,\mathcal T_t(xx^*)=\mathcal T_t(x)\mathcal T_t(x)^*\,, \forall t\geq 0}\,,
\end{gather*}
which is characterized in terms of the commutant $(\bigcup_{n\geq 0} \mathcal C_n)'$ of the following iterated commutators (cf. \cite{MR2446520})
\begin{align}\label{eq:iterated-commutators}
\mathcal C_n\colonequals\lrb{\delta_{H}^n(L_{\pm,\omega_k}),\delta_{H}^n(L_{\pm,\omega_k}^*)}_{k=0}^N=\lrb{\delta_{H}^n(Z_k),\delta_{H}^n(Z_k^*)}_{k=0}^N\,,
\end{align}
with $n\geq0$, where 
\begin{gather*}
\delta^0_H(X)=X\,,\quad \delta^1_H(X)=[H_{\mathrm{eff}},X]\,,\quad \delta_H^{n+1}(X)=[H_{\mathrm{eff}},\delta_H^n(X)]\,.
\end{gather*}

Denote by $\mathcal F(\mathcal T)$ the set of fixed points of the linear maps $\mathcal T_t$, given by
\begin{gather*}
\mathcal F(\mathcal T)\colonequals \lrb{x\in \mathcal {B(H)}\,:\, \mathcal T_t(x)=x\,,\quad \text{for all } t\geq 0}\,.
\end{gather*}

We omit the proof of the below theorem since it follows the same lines as the proof of \cite[Th.\,5.2]{AGQ2019D}.

\begin{thm}\label{th:ic-dbA}The commutators \eqref{eq:iterated-commutators} and the df-algebra of $\mathcal T$ satisfy
\begin{enumerate}
\item $\mathcal C_0'=\lrp{\mathcal C_0\cup \lrb{Z_kZ_k^*,Z_k^*Z_k}_{k=0}^N}'  \subset \mathcal F(\mathcal T)$.
\item $\mathcal C_0'\subset  \bigcap_{n\geq 1} \mathcal C'_n$.
\item\label{it3-ic-dbA} $\mathcal N(\mathcal T)=\mathcal C_0'\subset\mathcal F(\mathcal T)$. 
\end{enumerate}
\end{thm}

By virtue of Theorem~\ref{th:ic-dbA}.\eqref{it3-ic-dbA}, it follows that $\mathcal{N(T)}\subset \mathcal{F(T)}$ and equal if there exists a faithful invariant state in $\mathcal {B(H)}$ \cite[Sect.\,4]{AGQ2019D}. Additionally, Frigerio and Verri in \cite{MR661704,MR479136} assert that $\lim_{t\to\infty}\mathcal T_t(\eta)$ exists for any normal state $\eta\in \mathcal {B(H)}$. However, one has in view of Corollary~\ref{cor:support-invariant-state} that $(V\oplus\C\vk{0_{N+1}})^\perp$ is contained in the kernel of any invariant state. Hence, there is no faithful invariant state in $\mathcal {B(H)}$. 

The above reasoning requires restricting our discussion of evolution to hereditary subalgebras, where we can ensure the existence of a faithful invariant state. For instance, there exists by Theorem~\ref{th:fast-recurrent-subspace} an invariant state $\rho$ with $\ran\rho=\mathcal{R_L}$, i.e, it is faithful in the subalgebra $P_{\mathcal{R_L}}\mathcal{B(H)}P_{\mathcal{R_L}}$. Actually, any invariant state $\rho$ is faithful in $P_{\ran\rho}\mathcal{B(H)}P_{\ran\rho}$.

\begin{lem}\label{lem:range-condition-is}
If $\tau$ is a state supported in $V_1\ominus W$ then 
\begin{gather*}
\rho=c \sum_{n=0}^{N-1}e^{\sum_{j=0}^{n}\beta_j}Z^n\tau Z^{*n}\,,\quad\mbox{with}\quad c=\tr{\rho\abs Z_1}\,,
\end{gather*}
is an invariant state which satisfies
\begin{gather}\label{eq:range-condition-is}
\ran \rho=\bigoplus_{n=0}^{N-1}\ran Z^n\tau=\bigoplus_{n=0}^{N-1}Z^n\ran \tau\subset V\ominus W
\end{gather}
\end{lem}
\begin{proof}
It is clear from Lemma~\ref{rm:IS-inVmW} that $\rho$ is an invariant state supported in $V\ominus W$. Now, to show the first equality of \eqref{eq:range-condition-is} it is sufficient to prove that 
\begin{gather}\label{eq:induction-RC-is}
\ran Z^n\tau Z^{*n}=\ran Z^n\tau\,,\quad \mbox{for}\quad n=0,\dots,N-1\,,
\end{gather}
which is clear for $n=0$. Thereby, we may suppose that  \eqref{eq:induction-RC-is} is true for $n-1$. If $g\in \ran Z^n\tau$, with $g\neq0$, then $g=Z_nZ^{n-1}\tau v$, for some $v\in\dom \tau$ non-zero, and by hypothesis induction $g=Z_nZ^{n-1}\tau Z^{*n-1}w$, with $w\in V_n$ non-zero, since ${\rm supp}\, \rho\in V\ominus W$. We claim that $w\notin \ker Z_n$, otherwise Remark~\ref{eq:detail-balance} asserts that 
\begin{align*}
g=Z_nZ^{n-1}\tau Z^{*n-1}w=\frac{1}{ce^{\sum_{j=0}^{n-1}\beta_j}}Z_n\rho w=\frac{1}{ce^{\sum_{j=0}^{n}\beta_j}}\rho Z_nw=0\,,
\end{align*}
which is no posible. Thus, one has by Lemma~\ref{lem:properties-trasport-operator} that $w=Z_n^*u$, with $u\neq 0$ in $V_{n+1}$, and $g=Z^{n}\tau Z^{*n}u$. Hence, $\ran Z^n\tau\subset \ran Z^n\tau Z^{*n}$ which implies \eqref{eq:induction-RC-is}, since the other inclusion is straightforward. It is a simple matter to verify by containment that $\ran Z^n\tau=Z^n\ran \tau$, for $n=0,\dots,N-1$, which yields the second equality of \eqref{eq:range-condition-is}.
\end{proof}

We recall by Corollary~\ref{cor:convex-decomposition} that any invariant state is decomposable in three invariant states supported in $V\ominus W, W$ and $\C\vk{0_{N+1}}$, respectively, and Remark~\ref{rm:trivial-is} establish that every state supported in $W$ and $P_{N+1}$ are invariants. So, it is plausible to work only on hereditary subalgebras $P_{R}\mathcal{B(H)}P_{R}$, where $R$ is a subspace of $V\ominus W$. 

In what follows, $U$ represents a non-zero subspace in $V_1\ominus W$ and 
\begin{gather*}
U_Z\colonequals \bigoplus_{n=0}^{N-1}Z^n U\subset V\ominus W\,;\quad
\mathcal A_{U_Z}\colonequals P_{U_Z}\mathcal{B(H)}P_{U_Z}\,;\quad
\mathcal T_{U_Z,t}\colonequals P_{U_Z}\mathcal{T}_tP_{U_Z}\,,
\end{gather*}
where the hereditary semigroup $\mathcal T_{U_Z,t}$ acts on the hereditary subalgebra $\mathcal A_{U_Z}$.

\begin{rem}\label{rm:equality-hsg-states}
If a state $\rho$ belongs to $\mathcal A_{U_Z}$ then $\mathcal T_{U_Z,t}(\rho)=\mathcal T_t(\rho)$. Indeed, one simply checks that $\mathcal L(\rho)P_{U_Z}=P_{U_Z} \mathcal L(\rho)= \mathcal L(\rho)$. Thereby, $\rho$ is ivariant for $\mathcal T_{U_Z,t}$ if and only if it is for $\mathcal T_t$.
\end{rem}

\begin{cor}\label{cor:faithful-is}
There exists a faithful invariant state in $\mathcal A_{U_Z}$.
\end{cor}
\begin{proof}
Clearly, $\tau=\tr{P_U}^{-1}P_U$ is a state with $\ran \tau=U$ and by Lemma~\ref{lem:range-condition-is}, there exists an invariant state $\rho$ with $\ran \rho=U_Z$, which is faithful in $\mathcal A_{U_Z}$.
\end{proof}

\begin{rem}\label{rm:df-algebras-contaiment} The df-algebra $\mathcal{N}(\mathcal T_{U_Z})\subset \mathcal{N}(\mathcal T)$. Indeed, if $x\in \mathcal{N}(\mathcal T_{U_Z})$ then one has $x,x^*\in \mathcal A_{U_Z}\subset \mathcal B(\mathcal H)$. Taking into account Lemma~\ref{lem:properties-trasport-operator}, one simply computes that $Z_kU_Z,Z_k^*U_Z\subset U_Z$, for $k=1,\dots, N-1$, which implies that 
\begin{gather*}
\mathcal T_{t}(x^*x)=\mathcal T_{U_Z,t}(x^*x)=\mathcal T_{U_Z,t}(x)^*\mathcal T_{U_Z,t}(x)=\mathcal T_t(x)^*\mathcal T_t(x)\,,
\end{gather*}
as well as $\mathcal T_t(xx^*)=\mathcal T_t(x)\mathcal T_t(x)^*$, i.e., $x\in \mathcal{N}(\mathcal T)$.
\end{rem}

\begin{lem}\label{lem:her-dbalgebra} The df-algebra $\mathcal{N}(\mathcal T_{U_Z})$ is contained in $ \mathcal{F}(\mathcal T_{U_Z})$.
\end{lem}
\begin{proof}
If $\eta\in \mathcal{N}(\mathcal T_{U_Z})$ then by Remark~\ref{rm:df-algebras-contaiment}, it belongs to $\mathcal{N}(\mathcal T)$. It follows by Theorem~\ref{th:ic-dbA}.\eqref{it3-ic-dbA} that $\eta\in \mathcal C_0'=(\{Z_k,Z_k^*\}_{k=0}^N)'$, i.e., it commutes with $Z_k$, $Z_k^*$, for $k=0,\dots,N$, as well as $P_{\varphi_{0_1}},\abs Z_1,\dots,\abs Z_N,P_0,\dots,P_{N+1}$ (see properties of the transition operators in Section~\ref{s2}). Hence, from \eqref{eq:L-generator-rewrite} and since $\eta $ is supported in $U_Z$, it follows that $\eta$ is a fixed point of $\mathcal T_{U_Z,t}$, i.e,  $\eta\in \mathcal{F}(\mathcal T_{U_Z})$.
\end{proof}

Due to Corollary~\ref{cor:faithful-is} there exists a faithful invariant state in $\mathcal A_{U_Z}$ and as a consequence of Lemma~\ref{lem:her-dbalgebra}, one has that $\mathcal{N}(\mathcal T_{U_Z})=\mathcal{F}(\mathcal T_{U_Z})$ on $\mathcal A_{U_Z}$ (cf. \cite[Sect.\,4]{AGQ2019D}). Thereby, as a result of Frigerio and Verri \cite{MR661704,MR479136}, the following holds.

\begin{cor}\label{coro:existence-is-Avm}
If $\rho$ is an initial state in $\mathcal A_{U_Z}$, then $\lim_{t\to\infty}\mathcal T_{U_Z,t}(\rho)$ exists and is an invariant state in $\mathcal A_{U_Z}$.
\end{cor}

For an initial state $\rho\in\mathcal A_{U_Z}$, we write 
\begin{gather*}
\rho_\infty\colonequals \lim_{t\to\infty}\mathcal T_{U_Z,t}(\rho)\,.
\end{gather*}
which is invariant, by Corollary~\ref{coro:existence-is-Avm}. Remark~\ref{rm:equality-hsg-states} implies that $\rho_\infty=\lim_{t\to\infty}\mathcal T_t(\rho)$.
\begin{thm}\label{cor:evolution-rho}
For any initial state $\rho\in \mathcal A_{U_Z}$, there exists a unique state $\tau$ supported in $U$, such that 
\begin{gather}\label{eq:infty-state}
\rho_\infty=c_\rho \sum_{n=0}^{N-1}e^{\sum_{j=0}^{n}\beta_j}Z^n\tau Z^{*n}\,,\qquad (\beta_0=0)
\end{gather}
where $c_\rho=\tr{\rho_\infty\abs Z_1}$. Besides, 
\begin{gather}\label{eq:range-infty-state}
\ran \rho_\infty=\bigoplus_{n=0}^{N-1}Z^n\ran \tau\subset U_Z\,.
\end{gather}

\end{thm}
\begin{proof}
It follows from Corollary~\ref{coro:existence-is-Avm} that $\rho_\infty$ is an invariant state supported in $U_Z\subset V\ominus W$. Hence, by Lemma~\ref{rm:IS-inVmW} there exists a unique state $\tau$ supported in $V_1\ominus W$ such that satisfies
\eqref{eq:infty-state}. Note that $\ran \tau=P_1\ran \rho_\infty\subset P_1 U_Z=U$. Condition \eqref{eq:range-infty-state} readily follows from 
Lemma~\ref{lem:range-condition-is}.
\end{proof}

Equation \eqref{eq:infty-state} characterizes the long-time asymptotic behavior of states in $\mathcal A_{U_Z}$. In what follows, we will show a more explicit form of the evolution of states in this hereditary subalgebra.

From now on, we will assume that the subspaces $E_2,\dots, E_{N}\subset \mathcal H$ (see Section~\ref{s2}) satisfy the following \emph{dimension hypothesis} ($\mathrm{DH}$ for short):
\begin{align}\label{eq:condition-dimension}
\dim E_2&=\dots= \dim E_N\,,\quad \mbox{if $N\geq2$}
\end{align}
(the case $N=1$ is trivial and will be tackled in Subsection~\ref{sb51}). In such a case on $A_{U_Z}$, the equation \eqref{eq:L-generator-rewrite} turns into 
\begin{align}\label{eq:L-generator-rewritev2}
\begin{split}
\mathcal L(\rho)={}&\rho\lrp{\sum_{j=1}^{N-1}\eta_{-,\omega_j} P_j+\overline{\eta}_{+,\omega_j} P_{j+1}}+
\lrp{\sum_{j=1}^{N-1}\cc\eta_{-,\omega_j} P_j+{\eta}_{+,\omega_j} P_{j+1}}\rho  \\ & +
\sum_{k=1}^{N-1}\Gamma_{-,\omega_k} Z_k \rho Z_k^* + \Gamma_{+,\omega_k} Z_k^*\rho Z_k \,.
\end{split}
\end{align}

\begin{rem}\label{rm:normalization-constant} Under $\rm DH$, any invariant state $\rho\in A_{U_Z}$ satisfies
\begin{gather}\label{eq:normalization-constant}
c_\beta\colonequals\tr{\rho\abs Z_1}=\lrp{\sum_{n=0}^{N-1}e^{\sum_{j=0}^{n}\beta_j}}^{-1}\,. \qquad (\beta_0=0)
\end{gather}
Indeed, since $\rho$ has support in $U_Z\subset V\ominus W$, one has from Lemma~\ref{rm:IS-inVmW} that $\tr{Z^n\tau_\rho Z^{*n}}=1$, for $n=0,\dots,N-1$, and  
\begin{align*}
1=\tr{\rho}=\tr{\rho\abs Z_1}\sum_{n=0}^{N-1}e^{\sum_{j=0}^{n}\beta_j}\tr{Z^n\tau_\rho Z^{*n}}=\tr{\rho\abs Z_1}\sum_{n=0}^{N-1}e^{\sum_{j=0}^{n}\beta_j}\,,
\end{align*}
as required. Notably, the constant of Remark~\ref{rm:NC-by-spectrum} turns into $c=c_\beta$, since $\no{Z^nu_k}=1$ and $\sum_{k=1}^m\lambda_k=1$.
\end{rem}

\begin{lem}\label{lem:evolution-state-vk}
Under $\rm{DH}$, if $\rho_1,\rho_2,\dots,\rho_{N}$ are states supported in $U,ZU,\dots, Z^{N-1}U$, respectively,  then for $k=1,\dots, N$,
\begin{gather}\label{eq:evolution-states-a}
(\rho_k)_\infty=c_\beta \sum_{n=0}^{N-1}e^{\sum_{j=0}^{n}\beta_j}Z^nZ^{*k-1}\rho_k Z^{k-1} Z^{*n}\,,\qquad (\beta_0=0)
\end{gather}
with $\displaystyle \ran (\rho_k)_\infty=\bigoplus_{n=0}^{N-1}Z^n\ran Z^{*k-1}\rho_k Z^{k-1}\subset U_Z$.
\end{lem}
\begin{proof}
By abuse of notation, we let $\eta$ stand for $Z^{*k-1}\rho_k Z^{k-1}$, which from Lemma~\ref{lem:properties-trasport-operator} is a state supported in $U$. For $n\geq0$, consider $\eta_n=Z^n\eta Z^{*n}$, being $\eta_0=\eta$. Thus, it follows by \eqref{eq:L-generator-rewritev2} that 
\begin{align}\label{eq:gen-Tmtau}
\begin{split}
\mathcal L(\eta_0)&=-\Gamma_{-,\omega_{1}}\lrp{\eta_0 -\eta_1}\,,\\
\mathcal L(\eta_k)&=\Gamma_{+,\omega_{k}}\lrp{\eta_{k-1}-\eta_{k}}-\Gamma_{-,\omega_{k+1}}\lrp{\eta_{k}-\eta_{k+1}}\,,\quad k=1,\dots,N-2\,,\\
 \mathcal L(\eta_{N-1})&=\Gamma_{+,\omega_{N-1}}\lrp{\eta_{N-2} -\eta_{N-1}}\,.
\end{split}
\end{align}
For $k\in\N$, we claim that 
\begin{gather}\label{eq:powers-L}
\mathcal L^k(\eta)=\sum_{j=1}^{N-1}\alpha_{k,j}\lrp{\eta_{j-1}-\eta_{j}}\,,\quad \alpha_{k,j}\in\R\,.
\end{gather}
Indeed, \eqref{eq:powers-L} holds for $k=1$, due to \eqref{eq:gen-Tmtau}. So, we may suppose that \eqref{eq:powers-L} is true for $k$ and by virtue of \eqref{eq:gen-Tmtau}, one computes that
\begin{align*}
\mathcal L^{k+1}(\eta)=\sum_{j=1}^{N-1}\alpha_{k,j}\mathcal L\lrp{\eta_{j-1}-\eta_{j}}=\sum_{j=1}^{N-1}\alpha_{k+1,j}\lrp{\eta_{j-1}-\eta_{j}}\,.
\end{align*}
Thereby, since $\abs Z_1$ is a projection then it is bounded, and by \eqref{eq:powers-L}, it follows that $\abs Z_1\mathcal T_t(\eta)=\alpha_t\eta$, where $\alpha_t=\sum_{k\geq0} \frac{\alpha_{k,1}}{k!}t^k$, with $\alpha_{0,1}=1$. Note that $\eta\in\mathcal A_{U_Z}$ and by Theorem~\ref{cor:evolution-rho} there exists a unique state $\tau$ supported in $U$ such that $\eta_\infty,\tau$ satisfy \eqref{eq:infty-state} and by Remark~\ref{rm:normalization-constant},
\begin{gather}\label{eq:c-lim-at}
\lim_{t\to \infty}\alpha_t\eta=\abs Z_1\lim_{t\to \infty}\mathcal T_t(\eta)=\abs Z_1\eta_\infty=c_\beta\tau\,.
\end{gather}
So, $c_\beta=\tr{c_\beta\tau}=\tr{\lim_{t\to \infty}\alpha_t\eta}=\lim_{t\to \infty}\alpha_t$. Hence, $\eta=\tau$, which replacing in \eqref{eq:infty-state}, one gets \eqref{eq:evolution-states-a}. The above reasoning and \eqref{eq:range-infty-state} imply $\ran \rho_\infty=\bigoplus_{n=0}^{N-1}Z^n\ran \eta\subset U_Z$. 
\end{proof}

For $k=1,\dots,N$, if a state $\rho\in \mathcal A_{U_Z}$ satisfies $\rho P_k\neq0$, then $\frac{1}{\tr{\rho P_k}}P_k\rho P_k$ is a state supported in $Z^kU$. We say that $\frac{1}{\tr{\rho P_k}}P_k\rho P_k=0$ when $\rho P_k=0$.

\begin{thm}\label{th:evolution-states-vw}
Under $\rm DH$, if $\rho$ is an initial state in $\mathcal A_{U_Z}$, then 
\begin{gather}\label{eq:evolution-states-VW}
\rho_\infty=c_\beta \sum_{n=0}^{N-1}e^{\sum_{j=0}^{n}\beta_j}Z^n\eta Z^{*n}\,,\qquad (\beta_0=0)
\end{gather}
where $\displaystyle \eta=\sum_{k=0}^{N-1}Z^{*k}P_{k+1}\rho P_{k+1} Z^k$ is a state supported in $U$. Besides, 
\begin{gather}\label{eq:general-range-condition}
\ran \rho_\infty=\bigcup_{k=0}^{N-1}\bigoplus_{n=0}^{N-1}Z^{n}\ran Z^{*k}P_{k+1}\rho P_{k+1} Z^{k}\subset U_Z\,.
\end{gather}
\end{thm}
\begin{proof}
Since $\rho$ is supported in $U_Z$, then it follows that
\begin{gather}\label{eq:decomposition-rho}
\rho=\lrp{\sum_{k=1}^{N}P_k}\rho\lrp{\sum_{k=1}^{N}P_k}=\sum_{k=1}^{N}\alpha_k\rho_k+\sum_{h,k=1\atop h\neq k}^{N}P_h\rho P_k
\end{gather}
with $\alpha_k=\tr{\rho P_k}$ and $\rho_k=\alpha_k^{-1}P_k\rho P_k$, which is a state supported in $Z^{k-1}U$, for $k=1,\dots,N$. In this fashion, one obtains by virtue of Lemma~\ref{lem:evolution-state-vk} that 
\begin{gather}\label{eq:partk-rho}
(\rho_k)_{\infty}=c_\beta \sum_{n=0}^{N-1}e^{\sum_{j=0}^{n}\beta_j}Z^nZ^{*k-1}\rho_kZ^{k-1} Z^{*n}\,,\quad k=1,\dots,N\,,
\end{gather}
with 
\begin{align}\label{eq:rang-condition-rhok}
\ran(\rho_k)_{\infty}=\bigoplus_{n=0}^{N-1}Z^n\ran Z^{*k-1}\rho_k Z^{k-1}=\bigoplus_{n=0}^{N-1}Z^n\ran Z^{*k-1}P_k\rho P_k Z^{k-1}\subset U_Z\,.
\end{align}
 Thus, taking into account \eqref{eq:decomposition-rho} and \eqref{eq:partk-rho}, one computes that
\begin{align}\label{eq:evolution-states-VW-aux}
\begin{split}
\rho_\infty&=\sum_{k=1}^{N}\alpha_k(\rho_k)_{\infty}+\sum_{h,k=1\atop h\neq k}^{N}\lim_{t\to \infty}\mathcal T_t\lrp{P_h\rho P_k}\\
&=c_\beta \sum_{n=0}^{N-1}e^{\sum_{j=0}^{n}\beta_j}Z^n\eta Z^{*n}+\sum_{h,k=1\atop h\neq k}^{N}\lim_{t\to \infty}\mathcal T_t\lrp{P_h\rho P_k}\,.
\end{split}
\end{align}
It is clear that $\eta$ is a positive operator with support in $U$. Besides, one has that $\tr{\eta}=\tr{\rho\sum_{k=1}^NP_k}=\tr{\rho}=1$, i.e., $\eta$ is a state. Moreover, $\rho$ satisfies Theorem~\ref{cor:evolution-rho}, with $c_\rho=c_\beta$ (see Remark~\ref{rm:normalization-constant}), viz. $\rho_\infty$ satisfies \eqref{eq:infty-state} and $\tau=\eta$, which compared with \eqref{eq:evolution-states-VW-aux}, one concludes that $\sum_{h,k=1\atop h\neq k}^{N}\lim_{t\to \infty}\mathcal T_t\lrp{P_h\rho P_k}=0$, viz. \eqref{eq:evolution-states-VW}. Condition \eqref{eq:general-range-condition} follows from \eqref{eq:rang-condition-rhok} and the first equality \eqref{eq:evolution-states-VW-aux}.
\end{proof}
The following is straightforward from Theorem~\ref{cor:evolution-rho} and \eqref{eq:range-infty-state} of Theorem~\ref{th:evolution-states-vw}.
\begin{cor}\label{cor:attraction-domain}Under $\rm DH$, the attraction domain of the invariant state
\begin{gather*}
c_\beta \sum_{n=0}^{N-1}e^{\sum_{j=0}^{n}\beta_j}Z^n\eta Z^{*n}\,,\qquad (\beta_0=0)
\end{gather*}
where  $\eta$ is  state supported in $U$, consists solely of those initial states $\rho\in\mathcal A_{U_Z}$, for which 
\begin{align*}
\displaystyle \eta&=\sum_{k=0}^{N-1}Z^{*k}P_{k+1}\rho P_{k+1} Z^k\qquad\mbox{and}\\\bigoplus_{n=0}^{N-1}Z^n\ran \eta&= \bigcup_{k=0}^{N-1}\bigoplus_{n=0}^{N-1}Z^{n}\ran Z^{*k}P_{k+1}\rho P_{k+1} Z^{k}\subset U_Z\,.
\end{align*}
\end{cor}

\begin{rem}[Transport of states and energy] As a consequence of Theorem~\ref{th:evolution-states-vw}, the total probability of an initial state $\rho$ in $\mathcal A_{U_Z}$ is distributed in the limit when $t$ tends to infinity. Viz. the probability of $\rho_\infty$ in $Z^{k-1}U$ is 
\begin{align}\label{eq:probability-dist}
\tr{\rho_\infty P_k}=c_\beta e^{\sum_{j=0}^{k-1}\beta_j}\,,\qquad k=1,\dots,N\,,
\end{align}
which does not depend on the initial state $\rho$. Since we work under $\rm DH$ (see  \eqref{eq:condition-dimension}) and with states supported $U_Z\subset V\ominus W$, the effective Hamiltonian \eqref{eq:eff-Hamiltonian} turns into 
\begin{gather*}
H_{\rm{eff}}= \sum_{k=1}^{N-1}  \gamma_{-,\omega_k} P_k-
\gamma_{+,\omega_k}P_{k+1}\,.
\end{gather*}
Thereby, if initial states $\rho_1,\rho_2,\dots,\rho_N$ are supported in $U,ZU,\dots,Z^{N-1}U$, respectively, then for $k=1,\dots,N$, it follows by \eqref{eq:probability-dist} that 
\begin{align*}
\tr{\lrp{\rho_{k\infty}-\rho_k}H_{\rm{eff}}}= \gamma_{+,\omega_{k-1}}- \gamma_{-,\omega_k}+c_\beta\sum_{k=1}^{N-1}e^{\sum_{j=0}^{k-1}\beta_j}\lrp{ \gamma_{-,\omega_k}- \gamma_{+,\omega_k}e^{\beta_k}}\,,
\end{align*}
with $\gamma_{+,\omega_{0}}=\gamma_{-,\omega_N}=0$, which is independent of $\rho_k$. It seems that the degenerate open systems (with a degenerate reference Hamiltonian) are plausible for modeling effective quantum energy transfer in photosynthesis \cite{SGFGOC2011}.
\end{rem}

\section{Quantum photosynthesis models}\label{s7}\setcounter{equation}{0}

\subsection{Kozyrev-Volovich quantum photosynthesis model}\label{sb51}

The open quantum system with one energy level (see Fig. \ref{fig:graph-1e}) corresponds to Kozyrev and Volovich model \cite{Kozyrev2018} in the context of the stochastic limit approach of degenerate quantum open systems (c.f. \cite[Sect.\,3]{MR3860251} and \cite[Ex.\,3.2]{MR4107240}).

\begin{figure}[h]
\centering
\begin{tikzpicture}
  [scale=.75,auto=left,every node/.style={}]
  \node (n0) at (4.5,7.5) {$0_0$};
  \node (n1) at (0,6)  {$0_1$};
  \node (n2) at (3,6)  {$1_1$};
  \node (n3) at (6,6) {$\cdots$};
  \node (n4) at (9,6)  {$(n_1-1)_1$};
    \node (o0) at (4.5,4.5)  {$0_{2}$};
     \node (lM) at (13,7.5)  {$0$-level\,,\quad $E_0$};
      \node (l1) at (13,6)  {$1$-level\,,\quad $E_1$};
          \node (lm) at (13,4.5)  {$2$-level\,,\quad $E_{2}$};
  \foreach \from/\to in {n0/n1,n0/n2,n0/n3,n0/n4,o0/n1,o0/n2,o0/n3,o0/n4}
    \draw (\from) -- (\to);
\end{tikzpicture}
\caption{Graph of states and transitions with one energy level.}\label{fig:graph-1e}
\end{figure}
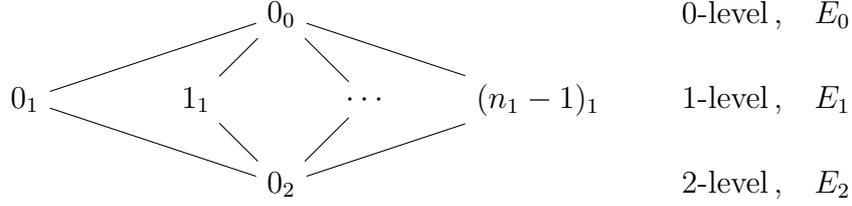

The transitions operators \eqref{eq:Gen-discreteFF} are given by 
\begin{gather*}
Z_0=\sqrt{n_1}\vk{\varphi_{0_1}}\vb{0_0}\quad \mbox{and}\quad Z_1=\vk{0_2}\vb{\varphi_{0_1}}\,.
\end{gather*}
Besides, the WCLT Markov generator ${\mathcal L}$ is 
\begin{align*}
\begin{split}
\mathcal L(\rho)={}&\rho\lrp{n_1\eta_{-,\omega_0}P_0+\lrp{n_1\overline{\eta}_{+,\omega_0}+\eta_{-,\omega_1}}P_{\varphi_{0_1}}+\overline{\eta}_{+,\omega_1} P_{0_2}}\\[2mm]&+
\lrp{n_1\cc\eta_{-,\omega_0}P_0+\lrp{n_1{\eta}_{+,\omega_0}+\cc\eta_{-,\omega_1}}P_{\varphi_{0_1}}+{\eta}_{+,\omega_1} P_{0_2}}\rho  \\[2mm] &+ n_1\Gamma_{+,\omega_0}\ip{\varphi_{0_1}}{\rho \varphi_{0_1}}P_{0_0}+n_1\Gamma_{-,\omega_0}\ip{\vk{0_0}}{\rho \vk{0_0}}P_{\varphi_{0_1}}\\[2mm]& +\Gamma_{-,\omega_1}\ip{\varphi_{0_1}}{\rho \varphi_{0_1}}P_{0_2}\,, 
\end{split}
\end{align*}
where $\eta_{\pm,\omega_k}=-\dfrac{\Gamma_{\pm,\omega_k}}{2}+i\gamma_{\pm,\omega_k}$, for $k=1,2$, with $\Gamma_{+,\omega_1}=0$. 

\textbf{Case $n_1>1$:} it follows from \eqref{eq:interaction-freeW} that $
W=\Span\lrb{\varphi_{a_1}}_{a=1}^{n_{1}-1}$ and by \eqref{eq:calV-harmonic}, one has
\begin{gather*}
V=\{0_0,\varphi_{0_1},0_2\}^\perp=W\,.
\end{gather*}
Therefore, due to Theorem~\ref{th:convex-decomposition-IS}, any invariant state is a convex combination of a state supported in $W$ and $P_{0_2}$. The invariant-extremal states are $P_{0_2}$ and $\vk{w}\vb{w}$, with $w$ a unit vector in $W$ (see Corollary~\ref{cor:invariant-extremal}). Furthermore, the fast recurrent subspace \eqref{eq:frs-RL} is $\mathcal{R_L}=W\oplus \C\vk{0_2}$ (see Theorem~\ref{th:fast-recurrent-subspace}).

\textbf{Case $n_1=1$:} in this case, one simply checks that $V=W=\{0\}$. Hence, $P_{0_2}$ is the only invariant state, which is invariant extremal, and $\mathcal{R_L}=\C\vk{0_2}$.

In both above cases, any state in the hereditary subalgebra $P_{\mathcal{R_L}}\mathcal{B(H)}P_{\mathcal{R_L}}$ is invariant. Hence, the analysis of the approach to equilibrium and attraction domains that we see in Section~\ref{s4} is simple in this subalgebra.

\subsection{Aref’eva-Volovich-Kozyrev quantum photosynthesis model}\label{sb52}
We frame the Aref’eva-Volovich-Kozyrev (briefly AVK) model \cite{MR3399653}, based on stochastic limit approach of degenerate quantum open systems (c.f. \cite{MR4107240}). This model is consistent with an open quantum system with two energy levels Fig. \ref{fig:graph-2e}.

\begin{figure}[h]
\centering
\begin{tikzpicture}
  [scale=.75,auto=left,every node/.style={}]
  \node (n0) at (4.5,7.5) {$0_0$};
  \node (n1) at (0,6)  {$0_1$};
  \node (n2) at (3,6)  {$1_1$};
  \node (n3) at (6,6) {$\cdots$};
  \node (n4) at (9,6)  {$(n_1-1)_1$};
   \node (m1) at (0,4.5)  {$0_2$};
   \node (m2) at (3,4.5)  {$1_2$};
   \node (m3) at (6,4.5)  {$\cdots$};
   \node (m4) at (9,4.5)  {$(n_2-1)_2$};
    \node (o0) at (4.5,3)  {$0_{3}$};
     \node (lM) at (13,7.5)  {$0$-level\,,\quad $E_0$};
      \node (l1) at (13,6)  {$1$-level\,,\quad $E_1$};
          \node (lm) at (13,4.5)  {$2$-level\,,\quad $E_{2}$};
          \node (l0) at (13,3)  {$3$-level\,,\quad $E_{3}$};
  \foreach \from/\to in {n0/n1,n0/n2,n0/n3,n0/n4,n1/m1,n1/m2,n1/m3,n2/m3,n2/m1,n2/m2,m1/o0,m2/o0,m3/o0,m4/o0,n4/m3,n4/m4,n3/m1,n3/m2,n3/m4,n3/m3}
    \draw (\from) -- (\to);
\end{tikzpicture}
\caption{Graph of states and transitions with two energy levels.}\label{fig:graph-2e}
\end{figure}
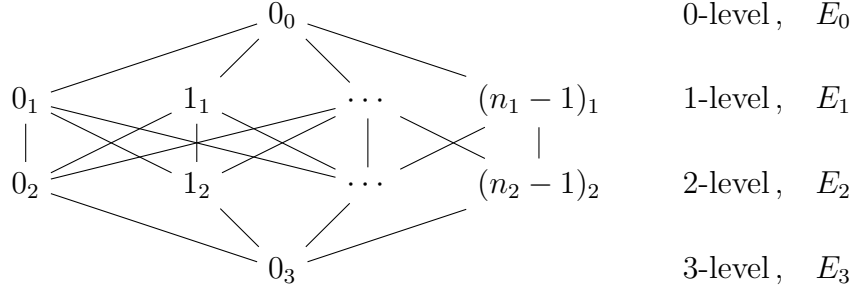

By virtue of \eqref{eq:Gen-discreteFF}, we only have three transitions operators 
\begin{gather*}
Z_0=\sqrt{n_1}\vk{\varphi_{0_1}}\vb{0_0} \,, \quad Z_1=\sum_{a=0}^{n_2-1}\vk{a_2}\vb{\varphi_{a_1}} \quad \mbox{and}\quad Z_2=\vk{0_3}\vb{\varphi_{0_2}}\,.
\end{gather*}

It is a simple matter to verify that the subspace \eqref{eq:calV-harmonic} is 
\begin{gather}\label{eq:V-for-N2}
V=\lrb{\vk{0_0}, \varphi_{0_1},Z_1^*\varphi_{0_2},\vk{0_2},\varphi_{0_2},\vk{0_3}}^\perp\,.
\end{gather}

Recall by \eqref{eq:interaction-freeW} that $W=\Span\lrb{\varphi_{a_1}}_{a=n_2}^{n_{1}-1}$ which is a subset of $V$. In the following, we will explicitly describe the elements of $V$ and the fast recurrence subspace \eqref{eq:frs-RL}. 

\begin{lem}\label{lem:V-description}
The subspace \eqref{eq:V-for-N2} satisfies
\begin{gather}\label{eq:v-explicitly}
V=\Span\lrb{\varphi_{a_1}-\varphi_{(a+1)_1},\varphi_{a_2}-\varphi_{(a+1)_2}}_{a=1}^{n_2-2}\oplus W\,.
\end{gather}
Thereby, \eqref{eq:V-for-N2} is decomposed in its levels by $V=V_1\oplus V_2$, where
\begin{gather*}
V_1=\Span\lrb{\varphi_{a_1}-\varphi_{(a+1)_1}}_{a=1}^{n_2-2}\oplus W\,;\qquad V_2=\Span\lrb{\varphi_{a_2}-\varphi_{(a+1)_2}}_{a=1}^{n_2-2}\,.
\end{gather*}
\end{lem}
\begin{proof}
If we denote the right-hand side of \eqref{eq:v-explicitly} by $M$, then it is simple to check that $\dim V=\dim M$. Thereby, we only need to show that $M\subset V$. Clearly, $\{\vk{0_0}, \varphi_{0_1},\varphi_{0_2},\vk{0_3}\}$ and $M$ are orthogonal. Besides, since $Z_1^*\varphi_{0_2}=n_2^{-1/2}\sum_{b=0}^{n_2-1}\varphi_{b_1}$, one has that $Z_1^*\varphi_{0_2}$ is orthogonal to $W$ as well as $V_2$, and for $a=1,\dots,n_2-2$,
\begin{gather}\label{eq:Zstarphi02-perpM}
\ip{\varphi_{a_1}-\varphi_{(a+1)_1}}{Z_1^*\varphi_{0_2}}=n_2^{-1/2}\sum_{b=0}^{n_2-1}\ip{\varphi_{a_1}-\varphi_{(a+1)_1}}{\varphi_{b_1}}=0\,,
\end{gather}
which implies $Z_1^*\varphi_{0_2}\perp M$. One obtains analogously to \eqref{eq:Zstarphi02-perpM} that $\vk{0_2}$ is orthogonal to $ M$, bering in mind that $\vk{0_2}=n_2^{-1/2}\sum_{b=0}^{n_2-1}\varphi_{b_2}$. Hence, $\lrb{\vk{0_0}, \varphi_{0_1},Z_1^*\varphi_{0_2},\vk{0_2},\varphi_{0_2},\vk{0_3}}\subset M^\perp$, i.e., $M\subset V$, as required.
\end{proof}
Lema~\ref{lem:V-description} and Theorem~\ref{th:fast-recurrent-subspace} give the following result.
\begin{thm}
The fast recurrent subspace in the AVK model is 
\begin{gather*}
\mathcal{R_L}=\Span\lrb{\varphi_{a_1}-\varphi_{(a+1)_1},\varphi_{a_2}-\varphi_{(a+1)_2}}_{a=1}^{n_2-2}\oplus W\oplus \C\vk{0_{3}}\,.
\end{gather*}
\end{thm}

Clearly, the AVK model is under $\rm DH$ condition \eqref{eq:condition-dimension}. Thereby, according to Theorem~\ref{th:convex-decomposition-IS} and Remark~\ref{rm:normalization-constant}, any state $\rho$ is invariant if and only if it is decomposed into a convex combination
\begin{gather*}
\rho=\frac\alpha{1+e^{\beta_1}}\lrp{\tau+e^{\beta_1}Z_1\tau Z_1^*}+\beta \eta +\lambda P_{3}\,, 
\end{gather*}
where $\alpha,\beta,\lambda\geq0$, with $\alpha+\beta+\lambda=1$,  and $\tau,\eta$ are states supported in the spaces $\Span\{\varphi_{a_1}-\varphi_{(a+1)_1}\}_{a=1}^{n_2-2},W$, respectively. 

Due to Corollary~\ref{cor:invariant-extremal}, any invariant-extremal state is characterized by being $P_{0_3}$, or $\vk w\vb w$ with $w$ a unit vector in $W$, or $\vk u\vb u+e^{\beta_1}Z_1\vk u\vb uZ_1^*$, viz. 
\begin{gather}\label{eq:ies-AKV}
\vk u\vb u+e^{\beta_1}\sum_{a=1}^{n_2-1}\abs{\ip{\varphi_{a_1}}{u}}^2\vk{a_2} \vb {a_2}\,,
\end{gather}
where $u\neq0$ belongs to $\Span\{\varphi_{a_1}-\varphi_{(a+1)_1}\}_{a=1}^{n_2-2}$, with $\no u=(1+e^{\beta_1})^{-1/2}$ (v.s. Remark~\ref{rm:normalization-constant}). Besides, Corollary~\ref{cor:ixe-eigenvalues} and Remark~\ref{rm:normalization-constant} assert that \eqref{eq:ies-AKV} has  
\begin{align*}
\mbox{non-zero eigenvalues }&\lrb{(1+e^{\beta_1})^{-1},e^{\beta_1}(1+e^{\beta_1})^{-1}}\,,\\\mbox{with respective eigenvectors } &\lrb{(1+e^{\beta_1})^{1/2}u,(1+e^{\beta_1})^{1/2}Z_1u}\,. 
\end{align*}

Now,  taking into account Theorem~\ref{th:evolution-states-vw}, for an initial state $\rho\in\mathcal A_{U_Z}$, where $U$ is a subspace in $\Span\{\varphi_{a_1}-\varphi_{(a+1)_1}\}_{a=1}^{n_2-2}$, one computes by \eqref{eq:evolution-states-VW} that
\begin{gather*}
\lim_{t\to\infty}\mathcal T_t(\rho)= \frac{1}{1+e^{\beta_1}}\lrp{P_1\rho P_1+Z_1^*\rho Z_1+e^{\beta_1}\lrp{P_2\rho P_2+Z_1\rho Z_1^*}}\,,
\end{gather*}
which is an invariant state (v.s. Corollary~\ref{coro:existence-is-Avm}) and satisfies
\begin{gather*}
\ran \lim_{t\to\infty}\mathcal T_t(\rho)=\bigcup_{k=0}^{1}\bigoplus_{n=0}^{1}Z_1^{n}\ran Z_1^{*k}P_{k+1}\rho P_{k+1} Z_1^{k}\subset U_Z\,.
\end{gather*}
E.g., for $j=1,2$ and  $u_j\in\Span\{\varphi_{a_j}-\varphi_{(a+1)_j}\}_{a=1}^{n_2-2}$, with $\no{u_j}=(1+e^{\beta_1})^{-1/2}$,
\begin{align*}
\lim_{t\to\infty}\mathcal T_t\lrp{\vk{u_1}\vb{u_1}}&=\vk{u_1}\vb{u_1}+e^{\beta_1}\sum_{a=1}^{n_2-1}\abs{\ip{\varphi_{a_1}}{u_1}}^2\vk{a_2} \vb {a_2}\,,\\
\lim_{t\to\infty}\mathcal T_t\lrp{\vk{u_2}\vb{u_2}}&=e^{\beta_1}\vk{u_2}\vb{u_2}+\sum_{a=1}^{n_2-1}\abs{\ip{a_2}{u_2}}^2\vk{\varphi_{a_1}} \vb {\varphi_{a_1}}\,.
\end{align*}
To conclude, in view of Corollary~\ref{cor:attraction-domain}, for a state $\eta$ with support in $U$, the attraction domain of the invariant state
\begin{gather*}
\frac{1}{1+e^{\beta_1}}\lrp{\eta+e^{\beta_1}Z_1\eta Z_1^*}\,,
\end{gather*}
is formed of those states $\rho\in \mathcal A_{U_Z}$, such that $\eta=\abs Z_1\rho \abs Z_1+Z_1^*\rho Z_1$ and 
\begin{gather*}
\ran \eta\oplus Z_1 \ran \eta=\bigcup_{k=0}^{1}\bigoplus_{n=0}^{1}Z_1^{n}\ran Z_1^{*k}P_{k+1}\rho P_{k+1} Z_1^{k}\subset U_Z\,.
\end{gather*}

\subsection*{Acknowledgments} This work was partially supported by UAM-PEAPDI 2023: ``Semigrupos cuánticos de Markov: Operadores de transición de niveles de energía y sus generalizaciones"  and Estancias Posdoctorales por México 2022, Id 2524732: ``El subespacio de recurrencia rápida en un modelo de transporte cuántico de energía''.

\def\cprime{$'$} \def\lfhook#1{\setbox0=\hbox{#1}{\ooalign{\hidewidth
  \lower1.5ex\hbox{'}\hidewidth\crcr\unhbox0}}} \def\cprime{$'$}
  \def\cprime{$'$} \def\cprime{$'$} \def\cprime{$'$} \def\cprime{$'$}
  \def\cprime{$'$} \def\cprime{$'$} \def\cprime{$'$}
  \def\lfhook#1{\setbox0=\hbox{#1}{\ooalign{\hidewidth
  \lower1.5ex\hbox{'}\hidewidth\crcr\unhbox0}}} \def\cprime{$'$}
  \def\cprime{$'$} \def\cprime{$'$} \def\cprime{$'$} \def\cprime{$'$}
  \def\cprime{$'$} \def\cprime{$'$}
\providecommand{\bysame}{\leavevmode\hbox to3em{\hrulefill}\thinspace}
\providecommand{\MR}{\relax\ifhmode\unskip\space\fi MR }
\providecommand{\MRhref}[2]{%
  \href{http://www.ams.org/mathscinet-getitem?mr=#1}{#2}
}
\providecommand{\href}[2]{#2}

\end{document}